\documentclass[journal]{IEEEtran}
\usepackage{circuitikz}
\usepackage{amsmath,amssymb}
\usepackage{cases}
\usepackage{graphicx,subfigure,epstopdf}
\usepackage[comma,numbers,square,sort&compress]{natbib}
\usepackage{epstopdf}
\usepackage{epic}
\usepackage{hyperref}
\usepackage{subfigure}
\usepackage{graphicx}
\usepackage{epstopdf}
\usepackage{epsfig}
\usepackage{arydshln}
\usepackage{enumerate}
\usepackage[utf8]{inputenc}
\usepackage[english]{babel}
\usepackage{fancyhdr}
\usepackage{lastpage}
\usepackage{color}
\usepackage{tikz}
\usepackage{dsfont}
\usepackage{changes}
\usepackage{booktabs}
 \usepackage{mathrsfs}
  \usepackage{caption}
\usepackage{changes}
\usetikzlibrary{arrows,shapes,chains}
\usetikzlibrary{calc}
\usetikzlibrary{positioning}

\newcommand{\col}{\mbox{col }}

\newtheorem{Remark}{Remark}
\newtheorem{Corollary}{Corollary}

\newenvironment{proof}{\noindent{\em Proof:\/}}{\hfill $\Box$\par}
\newtheorem{Theorem}{Theorem}
\newtheorem{Lemma}{Lemma}

\newtheorem{Assumption}{Assumption}
\newtheorem{Property}{Property}

\newcommand{\re}{\textnormal{Re}}
\newcommand{\im}{\textnormal{Im}}
\newcommand{\myb}{}
\newcommand{\myr}{}
\allowdisplaybreaks
\usepackage{lipsum}
\usepackage{mathtools}
\usepackage{cuted}
\begin{document}

\title{Distributed Time- and Event-Triggered Observers for Linear Systems: Non-Pathological Sampling and Inter-Event Dynamics}
\author{Shimin~Wang, Zhan Shu, Tongwen Chen\thanks{This work was supported by NSERC and an Alberta EDT Major Innovation Fund. Shimin~Wang, Zhan Shu and Tongwen Chen are with the Department of Electrical and Computer Engineering, University of Alberta, Edmonton, Alberta,  T6G~1H9, Canada. Emails: {\tt\small shimin.wang@queensu.ca, zshu1@ualberta.ca, tchen@ualberta.ca}}}
\maketitle
\begin{abstract}
For an autonomous linear time-invariant (LTI) system, a distributed observer with time-triggered periodic observations and event-triggered communication is proposed to estimate the state of the system.
It is shown that the sampling period is critical for the existence of desirable observers. A necessary and sufficient condition is established to give all feasible sampling periods that lead to convergent error dynamics and characterize delicate relationships among sampling periods, topologies, and system matrices.
An event-triggering mechanism based on locally sampled data is designed to regulate the communication among agents, and the convergence of the estimation errors under the mechanism holds for a class of positive and convergent triggering functions, which include the commonly used exponential function as a special case.
The mixed time- and event-triggered architecture naturally excludes the existence of Zeno behavior as the system updates at discrete instants.
When the triggering function is bounded by exponential functions, analytical characterization of the relationship among sampling, event triggering, and inter-event behaviour is established.
%
%
Finally, several examples are provided to illustrate the effectiveness and merits of the theoretical results.
\end{abstract}

\section{Introduction}
{\myr Cooperative control of distributed systems has been extensively studied over decades, e.g., coupled oscillators \citep*{proskurnikov2016synchronization,poveda2019hybrid}, consensus \citep*{ren2005consensus,priscoli2015leader}, and formation control \citep*{mou2015target, vidal2003formation}.
An essential class of cooperative control problems is the {leader-following consensus problem}, which aims to drive the state of agents to synchronize with an autonomous system called the leader.
Typically, the leader is described by the following autonomous differential equation
\begin{align}\label{leader}
\dot{v}(t) &= Sv(t),
\end{align}
where $S\in \mathds{R}^{n\times n}$ is the system matrix of the leader; $v(t)\in \mathds{R}^{n}$ is the state of the leader representing the reference signal to be tracked.
A compulsory difficulty in overcoming the
the problem is that the state of the leader is not available for each follower or some nodes.
 As a result, the problem naturally raises the distributed state estimation problem, which aims to design the distributed observer asymptotically estimate the state of the leader system using local interaction with its neighbors \cite{wu2021design}.}

{\myr The distributed state estimation problem took shape
in \cite{olfati2007distributed,su2011cooperative} and was extended in the more general condition in \cite{wang2017distributed,mitra2018distributed,han2018simple,kim2019completely}.
To be more precise, a distributed estimation algorithm was designed to estimate the state of the leader over the undirected graph in \cite{olfati2007distributed} under the assumption that each agent can observe the leader's state independently.
The distributed observer over the general directed graph was proposed in \cite{su2011cooperative} to solve the cooperative output regulation problem. In this case, not all the followers can access the leader's signal $v(t)$, which can be viewed as semi-jointly observable distributed systems.
%
%
Another type of distributed observer was constructed in \cite{mitra2018distributed} under locally jointly observable systems assumption in the sense that each follower can only access the locally partial state of the leader and effectively communicates with its neighbors to estimate the state of the leader cooperatively.
%
%
%
It is worth mentioning that \cite{wang2017distributed,han2018simple,kim2019completely} designed the distributed observer under the jointly observable assumption that each follower can access a limited amount of state about the leader system, and each distributed sensor node can achieve an estimation of the leader's entire state.
Admittedly, the jointly observable assumption is the mildest possible restriction on the leader system as it allows the leader's state to be unobservable to each follower. In the meantime, the follower can still reconstruct the state by locally continuously interacting with their neighbor.
More recently, the study of the distributed observer has turned into a nonlinear leader over the jointly observable assumption in \cite{wu2021design}, and semi-jointly observable assumption in \cite{liu2019cooperative}. }

%
%
%
%
%
%
%
%
%

Communication constraints naturally exist in distributed systems, posing grand challenges to control analysis and design. Most existing cooperative control laws and protocols rely on continuous sensing and communication among agents, which is over-idealized and hard to implement.
Therefore, more realistic designs based on intermittent signals are desired \citep*{chen2012optimal}.
Time-triggered and event-triggered sampling schemes are two commonly used approaches in practice to generate intermittent signals for estimation and control.
With the time-triggered approach, continuous-time signals are often sampled periodically, resulting in periodic observations and control updates \citep*{chen1991h}.
By contrast, with the event-triggered approaches, continuous-time signals are sampled according to prescribed or adaptive triggering conditions, leading to irregular observations and control updates \citep*{astrom2002comparison,meng2012optimal,borgers2014event}, which may significantly reduce unnecessary consumption of resources.
In this regard, cooperative control using the event-triggered approach has received a great deal of attention \citep*{yu2020zeno}. 
Moreover, event-triggered robust cooperative output regulation for a class of linear minimum-phase distributed systems was considered in \cite{liu2017event} by using an augmented system approach, and a more complex nonlinear version was studied in \cite{liu2018cooperative} by employing a distributed internal model design.

{\myr In particular, the event-triggered approach has been applied to the distributed state estimation problem for both linear and nonlinear distributed systems \citep*{liu2017event,hu2017cooperative,hu2018event,liu2018cooperative}.
For example, \cite{hu2017cooperative} proposed an event-triggered distributed observer under the semi-jointly observable assumption to achieve cooperative output regulation by utilizing a self-triggering strategy, and the problem over switching networks was considered in \cite{hu2018event}.
In addition, an event-triggered distributed observer subject to combined output observable conditions is firstly designed in \cite{deng2021distributed} for distributed tracking problems for uncertain nonlinear distributed systems.
However, to verify whether events should occur or not, the results mentioned above require continuous and real-time monitoring of the neighbors' states, which contradicts the original purpose of introducing event-based control as a means of reducing communication in distributed systems\cite{meng2013event}.}

Recent efforts have focused on designing triggering mechanisms with localized information to reduce the dependence on real-time information of neighbor agents. In \cite{liu2017distributed}, an open-loop estimator was designed for each agent to estimate the neighbors' states and achieve leaderless synchronization. {\myr This estimator-based approach was further used to solve the event-triggered distributed tracking problem for the linear distributed systems in \cite{cheng2017event} and the nonlinear distributed systems in \cite{deng2021distributed}.} In \cite{meng2018event}, a distributed observer involving multiple predictors and an edge-based triggering mechanism was proposed to remove the real-time inter-agent communication, and a similar approach was applied to cooperative output regulation with bounded inputs in \cite{cheng2019coordinated}.

All the aforementioned results are based on continuous-time signals of agents, and the triggering condition has to be monitored and verified continuously. However, most control algorithms and communication protocols are implemented with digital devices, resulting in clock-driven measurements and updates. Therefore, much effort has been devoted to developing various event-triggered strategies based on discrete-time signals of agents, among which the periodic event-triggered strategy has received much attention \citep*{heemels2012periodic, meng2013event, guo2014distributed,garcia2016periodic}. Very recently, the periodic event-triggered strategy has been employed for the distributed state estimation problem of linear distributed systems in \cite{zheng2020periodic}. With the periodic event-triggered strategy, the triggering condition is checked at discrete sampling instants, making it suitable for digital implementation.
Also, it naturally excludes Zeno behavior, which is undesirable due to its impracticability \citep*{heemels2012periodic}.
Besides, the inter-event steps with the periodic event-triggered strategy are multiples of the sampling period, facilitating related protocol design.
{\myr Just as \cite{liu2015event} pointed out, the numerical simulation's step-length guarantees strictly positive minimum inter-event times and excludes Zeno behavior. While \cite{liu2015event} still observed that the inter-event times converge to the step-length of the numerical simulation in the finite time. Hence, the periodic event-triggered strategy is of practical interest and deserves more effort for distributed systems. }
%

%

%
{\myr In light of the pathological sampling, a significant design issue of the sampled-data systems (\cite[Section~3.2]{chen2012optimal}) has rarely been treated in the existing literature related to distributed systems.
Pathological sampling means the discretized version of a stabilizable/detectable continuous-time plant is not itself stabilizable and detectable\citep*{middleton1995non}.
A typical phenomenon of pathological sampling is called aliasing, or frequency folding \cite{middleton1990digital}, which means that the observer reconstructs a different signal from the sampled value leading to misinterpretation.}
{\myr As in \cite{castillo1997regulation,lawrence2001output}, another negative effect of pathological sampling is that the solvability of the sampled output regulation cannot be achieved, such as \cite{wang2021robust}.}
Although some researchers have noticed the impact of the sampling period on consensus behavior of second-order systems \cite{yu2011second,huang2016some}, general results on the sampling period and its relationship with stability and performance of distributed systems are pretty limited, {\myr especially non-pathological sampling for the distributed systems.
%
In addition, the stability of the continuous-time distributed systems only relies on the minimum real part among the eigenvalues of the Laplacian matrix in \cite{su2011cooperative}.
Only choosing the sampling period for each subsystem to guarantee controllability and observability is not enough for the globally distributed systems to achieve consensus.
Moreover, the pathological sampling interval lies on the entire positive real number axis in discrete form, and the Lebesgue measure could be zero.}
The sampling period in most cases is either selected as a sufficiently small number or restricted to some values on a fixed interval determined by a set of linear matrix inequalities or the spectral radius/norm of a specific graph-induced matrix \citep*{meng2013event, guo2014distributed, liu2019leader, zheng2020periodic}, {\myr which are merely sufficient conditions and partially reliable solutions for the distributed systems}.
%
%
Apart from this, the interaction between periodic sampling and event triggering is poorly understood, and their impact on the inter-event dynamics remains unclear. {\myr Not to mention that how the inter-event times behave in the continuous case is still unknown, as pointed out in \cite{postoyan2019inter}}.

Motivated by the studies mentioned above, distributed state estimation problem of linear distributed systems over the {\myr semi-jointly observable condition} is investigated in this paper using \emph{mixed time- and event-triggered observers} approach. The main contributions are summarized as follows:
\begin{enumerate}
  \item Mixed time- and event-triggered observers are proposed to estimate the state of the leader with discontinuous monitoring and localized event triggering {\myr over the semi-jointly observable assumption}. {\myr{The superposition of
distributed event-triggered sampling and time-based sampling is studied extensively.}}
  \item The interactions among sampling periods, topologies, reference signals, and related observability has been fully revealed, and all non-pathological / pathological sampling periods are given. {\myr{ Other than the continuous-time case in \cite{su2011cooperative}, we show that the image parts among the eigenvalues of the Laplacian matrix also play essential roles in the convergence of time-triggered distributed observers.}}
  \item {\myr The interplay between the phrase of the graph and non-pathological sampling periods is illustrated by using two exceptional
cases. The lower bound and existence of optimal solution related to the convergence rate of the time-triggered distributed observers are exhibited.}
  \item The convergence of the error dynamics holds for a class of convergent triggering functions, including the commonly used exponential function as a particular case.
  \item The mixed time- and event-triggered architecture naturally excludes Zeno behavior.  For a particular class of triggering functions bounded by exponential functions, inter-event behaviors have been analyzed, and their links with period sampling and event triggering have been revealed.
\end{enumerate}

The rest of this paper is organized as follows: {\myr In Section \ref{section2}, a pair of time-triggered and event-triggered observers is designed, the distributed state estimation is formulated, and a standard assumption and some technique lemmas are introduced.} Section \ref{mainresults} is devoted to analyzing related stability. The interaction between periodic sampling and event triggering is revealed in Section \ref{intereventstep}.
Examples are provided in Section \ref{section6} to show the effectiveness and efficiency of the proposed approach, followed by conclusions given in Section~\ref{conlu}.

\textbf{Notation:} $\mathds{R}$ ($\mathds{R}^{+}$) and $\mathds{C}$ are the set of (positive) real numbers and the set of complex numbers, respectively. $\mathds{N}$ denotes all the natural numbers. $\mathds{Z}$($\mathds{Z}^{+}$) is the set of (positive) integers. $I_n$ denotes the $n\times n$ identity matrix.  For a collection of vectors $b_i\in \mathds{R}^{n_i \times p}$, $\col(b_1,\dots,b_m)\triangleq\big[b_1^T, \cdots, b_m^T\big]^T.$ For a square matrix $S\in \mathds{R}^{n\times n}$,  $\lambda_q$, $q=1,2,\dots,n$, represent its eigenvalues. For a complex number $z\in \mathds{C}$, $\operatorname{Re}(z)$ and $\operatorname{Im}(z)$ denote the real and imaginary parts of $z$, respectively; $\operatorname{Arg}\left(z\right)\in(-\pi,\pi]$ denotes the principal value of the argument of $z$; $|z|$ denotes the modulus of $z$; $z^*$ denotes the complex conjugate of $z$. For a matrix $B\in \mathds{R}^{m\times n}$, $\|B\|$ stands for the 2-norm of $B$. {\myr $\otimes$ represents the Kronecker product of matrices, and for any matrices $A$, $B$, $C$, $D$ of conformable dimensions, it has the following properties
\begin{align*}
(A\otimes B)(C\otimes D)=&(AC)\otimes (BD),\\
(A+B)\otimes(C+D)=&(A\otimes C)+(A\otimes D)\\
&+(B\otimes C)+(B\otimes D).
\end{align*}
}

\section{Objective Formulation}\label{section2}

\subsection{Graph theory}

The distributed system considered in this paper involves a leader and $N$ followers.
They are connected through a network topology described by a digraph $\mathcal{G}\triangleq\left(\mathcal{V},%
\mathcal{E}\right)$ with $\mathcal{V}=\{0,\dots,N\}$ and $\mathcal{E}\subseteq\mathcal{V}^2$, where $\mathcal{V}$ and $\mathcal{E}$ represent the sets of vertices and edges, respectively. Throughout the paper, vertex $0$ is associated with
the leader and vertex $i\in\mathcal{V}_{F}\triangleq\{1,2,\dots,N \}$ is associated with each follower.
For $i\in\mathcal{V}$ and $j\in\mathcal{V}_{F}$, the ordered pair $(i,j) \in {\mathcal{E}}$ represents that there is a directional communication link from agent $i$ to agent $j$. $\mathcal{%
N}_i\triangleq\{j|(j,i)\in \mathcal{E}\}$ denotes all the neighbors of agent
$i$. The Laplacian of a digraph is represented by $L=\left[l_{ij}\right]\in\mathds{R}^{N+1}$, where $l_{i i}=-\sum_{j=1}^{N} l_{i j}$, $l_{ij}=-1$ for $i \neq j$ if $(j, i) \in \mathcal{E}$, and $l_{ij}=0$ for $i \neq j$ otherwise. $H\in\mathds{R}^{N}$ denotes the leader-following matrix obtained by deleting the first row and column of ${L}$ \citep*{su2011cooperative,yan2018new}. More details on the graph theory can be found in \cite{godsil2013algebraic}.

\subsection{Mixed time- and event-triggered observers}

\begin{figure}[htp]
\centering
\includegraphics[trim=165 586 -90 125,clip,scale=0.9]{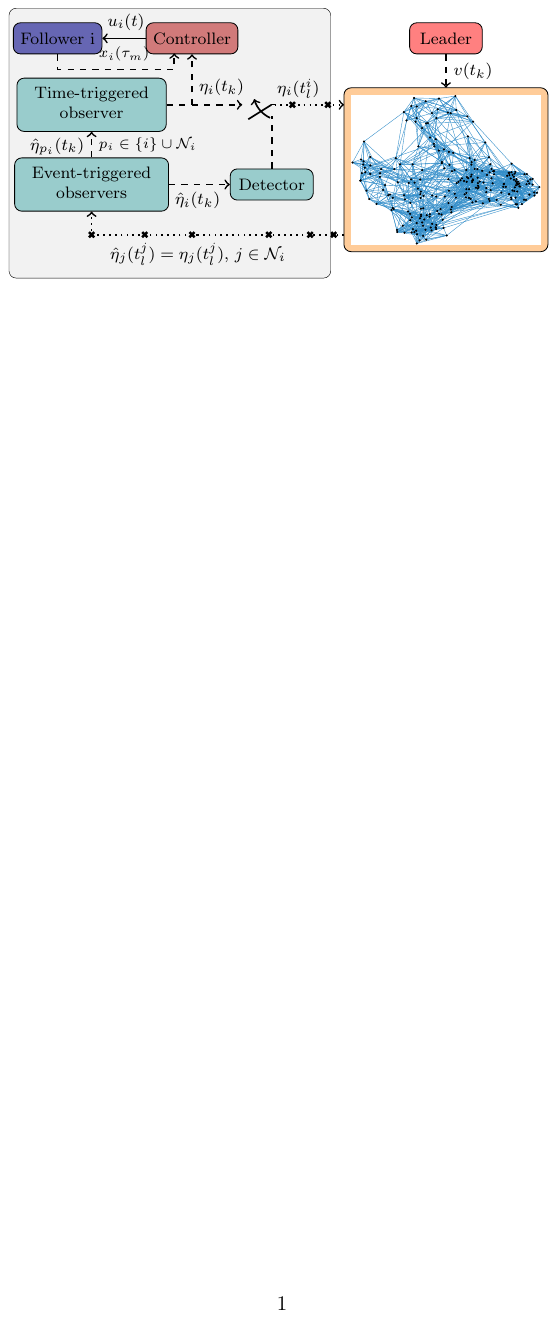}
\caption{Schematic of mixed time- and event-triggered observers}\label{figschematic}
\end{figure}

{\myr The dynamics of the followers will be omitted in this paper. We will mainly focus on designing the mixed time- and the event-triggered distributed observer. The mixed time- and the event-triggered observer is a dynamic compensator, which simultaneously estimates the leader's state by using intermittent local interaction such that the follower can synthesize control law to achieve the control goal such as consensus or formation pattern. The design schematic is shown in Figure \ref{figschematic}, and the time-triggered observer and the event-triggered observer are proposed to estimate the leader's state $v(t)$ without continuous information exchange.}

As the state of the leader is assumed to be accessible to some followers, we introduce the following time-triggered observer for each follower to estimate the state of the leader
\begin{equation}
\dot{\eta}_i(t)=S\eta_i(t) + \mu\sum\nolimits_{j\in\mathcal{N}_i}{(\hat{\eta}_j(t)-\hat{\eta}_i(t))},\label{compensator}
\end{equation}
where $\mu>0$ is the observer gain to be determined; $\hat{\eta}_0(t)=v(t)$; $\eta_i(t)\in \mathds{R}^n$, $i\in\mathcal{V}_{F}$, is the estimation of $v(t)$; $\hat{\eta}_i(t)$ and $\hat{\eta}_j(t)$, $j\in\mathcal{N}_i$ are the auxiliary estimation signals
of $\eta_i(t)$ and $\eta_j(t)$, $j\in \mathcal{N}_i$, respectively. For each follower $i$, the auxiliary estimation signals $\hat{\eta}_{p_i}(t)$, $p_i\in \left\{i\right\}\cup\mathcal{N}_i$, are updated periodically with a step size $h>0$, and the zero-order holder is used to maintain the signals between sampling instants, namely,
\begin{align*}
\hat{\eta}_{p_i}\left(t\right) & =\hat{\eta}_{p_i}\left(t_{k}\right),\;t\in\left[t_{k}~t_{k+1}\right),~k\in \mathds{N},\\
h & =t_{k+1}-t_{k}.
\end{align*}
The discrete-time auxiliary estimation signals $\hat{\eta}_{p_i}\left(t_{k}\right)$, $p_i\in \left\{i\right\}\cup\mathcal{N}_i$, of each follower $i$ are generated by an auxiliary observer with event-triggered communication:
 \begin{align}\label{exosystem1}
 \hat{\eta}_{p_i}(t_{k+1})&=e^{Sh}\hat{\eta}_{p_i}(t_k),~t_k\neq t_{l}^{p_i},~l\in \mathds{Z}^{+},\nonumber\\
\hat{\eta}_{p_i}(t_{l}^{p_i})&=\eta_{p_i}(t_{l}^{p_i}),~t_k= t_{l}^{p_i},~l\in \mathds{Z}^{+},
 \end{align}
where $t_{l}^{p_i}$ denotes the $l$-th triggering instant of follower $p_i$. Each follower $i$ is equipped with such a time-triggered observer on $\eta_{i}$ and an event-triggered observer on $\hat{\eta}_{p_i}$, $p_i\in \left\{i\right\}\cup\mathcal{N}_i$. At each triggering instant $t_{l}^{i}$, follower $i$ broadcasts its estimation $\eta_{i}\left(t_{l}^{i}\right)$, and related followers update their auxiliary observers accordingly. The sequence of triggering instants $t_{l}^{i}$, $l\in \mathds{Z}^{+}$ for each follower $i$ is generated by the following event-triggered mechanism
\begin{subequations}\label{setrigger}
\begin{align}
t_{l+1}^{i}&=t_{l}^{i}+s^i_{l} h,\label{setrigger1}\\
s^i_{l}&=\min\left\{r\in\mathds{Z}^{+}\Big|
           \begin{array}{c}
             \|\bar{\eta}_i(t_k)\|> f_i(t_k), \\
             t_k=t_{l}^{i}+rh. \\
           \end{array}\right\},\label{setrigger2}
\end{align}
\end{subequations}
where $s^i_{l}$ is the inter-event step of follower $i$; $\bar{\eta}_i(t)\triangleq\eta_i(t)-\hat{\eta}_{i}(t)$ is the \emph{absolute auxiliary error} of follower $i$; $f_i(t_k)$ for each $i\in\mathcal{V}_{F}$ is a positive bounded function and \emph{ converges to zero} as $t_k\rightarrow\infty$, e.g., $f_i(t_k)=\frac{\sigma_i}{\alpha_i t_k+1}$ or $f_i(t_k)=\sigma_i\ln \left(1+ e^{-\alpha_i t_k }\right)$.
With the mechanism in (\ref{setrigger}), triggering events only occur at sampling/updating instants, that is, $\{t_{0}^{i},t_{1}^{i},\cdots\}\subseteq \{t_0,t_1,t_2,\cdots\}$, and the set of inter-event time $\{t_{l+1}^{i}-t_{l}^{i},~l\in \mathds{Z}^{+},~i\in\mathcal{V}_F\}$ is lower bounded by the sampling period $h$.

$\eta_i(t)\in \mathds{R}^n$, $i\in\mathcal{V}_{F}$, is the estimation of $v(t)$; $\hat{\eta}_i(t)$ and $\hat{\eta}_j(t)$, $j\in\mathcal{N}_i$ are the auxiliary estimation signals
of $\eta_i(t)$ and $\eta_j(t)$, $j\in \mathcal{N}_i$, respectively. For each follower $i$, the auxiliary estimation signals $\hat{\eta}_{p_i}(t)$, $p_i\in \left\{i\right\}\cup\mathcal{N}_i$, are updated periodically with a step size $h>0$. Each follower $i$ is equipped with such a time-triggered observer on $\eta_{i}$ and an event-triggered observer on $\hat{\eta}_{p_i}$, $p_i\in \left\{i\right\}\cup\mathcal{N}_i$.

{\myr \begin{Remark}  Theorem \ref{radio} can be directly applied to the leader-following consensus problem with system \eqref{leader}  as the leader and the following system
\begin{equation}
\dot{x}_i(t)=Sx_i(t)+u_i(t),~~~~~~~~i=1,\cdots,N \nonumber
\end{equation}
as $N$ follower subsystems. The leader-following consensus problem can be solved by following the zero-
order holder control law $$u_i(t)=\mu\sum\nolimits_{j\in\mathcal{N}_i}{(x_j(t_k)-x_i(t_k))},$$
where $x_0=v$, the positive $\mu$ exist if the sampling period $h$ satisfies all the conditions in the Theorem \ref{radio}.
\end{Remark}}

{\myr
The objective of this paper is to design a mixed time- and event-triggered
distributed observer to estimate the state of system \eqref{leader} in the
sense that the estimation state $\eta(t)\in \mathds{R}^{n}$ of each $i$th local follower's
observer converges to the state $v(t)$, i.e.,
$$\lim_{t\rightarrow \infty}(\eta(t)-v(t))=0.$$
}

{\myr
%
Another objective of this paper is that we want to reveal the interaction between periodic
sampling and event triggering and their
impact on inter-event dynamics.}

To move on, we need these assumptions as follows:

\begin{Assumption}\label{ass0}$\mathcal{G}$ contains a spanning tree with vertex 0 as the
root.
\end{Assumption}
Assumption \ref{ass0} is required for distributed observer design. Under Assumption \ref{ass0}, all eigenvalues of $H$ have positive real parts, that is, $\re\left(\lambda_i\left(H\right)\right)>0$, for $i\in\mathcal{V}_{F}$. According to the real parts of the eigenvalues of $S$, a partition of the set $\mathcal{Q}=\{1,\cdots,n\}$ is constructed as follows:
\begin{subequations}\label{partitionof1ton}
\begin{align}
\mathcal{Q}_1&=\{q|\operatorname{Re}(\lambda_q)>0, q\in\mathcal{Q}\},\\
\mathcal{Q}_2&=\{q|\operatorname{Re}(\lambda_q)<0, q\in\mathcal{Q}\},\\
\mathcal{Q}_3&=\{q|\operatorname{Re}(\lambda_q)=0, \operatorname{Im}(\lambda_q)\neq0, q\in\mathcal{Q}\},\\
\mathcal{Q}_4&=\{q|\operatorname{Re}(\lambda_q)=0, \operatorname{Im}(\lambda_q)=0, q\in\mathcal{Q}\}.
\end{align}
\end{subequations}
\subsection{Technical lemmas}
In the remaining of this paper, $H$ and $S$ in $\lambda_i\left(H\right)$ and $\lambda_q\left(S\right)$ are omitted for notational simplicity unless specifically stated, as they can be deduced from the subscripts $i$ and $q$. Also, suppose that the polar form of $\lambda_i$ and $\lambda_q$ can be written as
\begin{equation*}
 \lambda_i=|\lambda_i|e^{\theta_i j} ~~~~\textnormal{and}~~~~ \lambda_q=|\lambda_q|e^{\theta_qj},
\end{equation*}
where $\theta_i=\operatorname{Arg}\left(\lambda_i\right)$ and $\theta_q=\operatorname{Arg}\left(\lambda_q\right)$ for $q\in\mathcal{Q}$ and $i\in\mathcal{V}_F$. To facilitate the proof of the main results, for any $q\in\mathcal{Q}$ and $i\in\mathcal{V}_F$, four quantities are introduced:
\begin{subequations}\label{UVPT}
\begin{align}
U_{q}(h)&\triangleq e^{\re(\lambda_q)h}-\cos(\im(\lambda_q)h),\\
V_{q}(h)&\triangleq\sin(\im(\lambda_q)h),\\
\phi_{q}(h)&\triangleq\operatorname{Arg}\left(U_{q}(h)+jV_{q}(h)\right),\\
\psi_{i,q}(h)&\triangleq\theta_i+\phi_{q}(h)-\theta_q.
\end{align}
\end{subequations}
$\phi_{q}(h)$ is well-defined when $U_{q}(h)+jV_{q}(h)\neq0$. If this is not the case, then $\phi_{q}(h)$ will be determined through some limiting processes, e.g., $\lim_{h \rightarrow0}\operatorname{Arg}\left(U_{q}(h)+jV_{q}(h)\right)$ or $\lim_{\lambda_q \rightarrow0}\operatorname{Arg}\left(U_{q}(h)+jV_{q}(h)\right)$. The lemma below gives some important properties of these quantities, and its proof is provided in Appendix A.
\begin{Lemma}\label{lemmauvpt} For $U_{q}(h)$, $V_{q}(h)$, $\theta_q$, and $\phi_{q}(h)$ defined in \eqref{UVPT}, the following relationships hold:
\begin{enumerate}
    \item $\lim_{\lambda_q\rightarrow0}\frac{V_{q}^2(h)+U_{q}^2(h)}{|\lambda_q|^2}=h^2$,
    \item $\lim_{\lambda_q\rightarrow0}(\phi_{q}(h)-\theta_q)=0$,
    \item $\lim_{h \rightarrow0}\frac{U_{q}^2(h)+V_{q}^2(h)}{h^2}=|\lambda_q|^2$,
    \item $\lim_{h\rightarrow0}(\phi_{q}(h)-\theta_q)=0$.
\end{enumerate}
\end{Lemma}
We end this section by presenting a technical lemma \citep*[Lemma~1]{huang2016cooperative}, which will be used for stability analysis.
\begin{Lemma}\label{lemmaaymmk}Consider the following discrete-time system
\begin{equation*}
x(t_{k+1})=Fx(t_k)+G_1(t_k)x(t_k)+G_2(t_k),
\end{equation*}
where $x(t_k)\in \mathds{R}^{n}$, $F\in \mathds{R}^{n\times n}$, $G_1(t_k)\in \mathds{R}^{n\times n}$, and $G_2(t_k)\in \mathds{R}^{n}$. Suppose that $F$ is Schur, $G_1(t_k)$ and $G_2(t_k)$ are well-defined\footnote{\myr{The function $f: A\mapsto B$ is well defined if for each $x\in A$ there is a unique $b\in B$ with $f(a)=b$.}} for all $k\in \mathds{Z}^+$. Then, if $\lim\limits_{t_k\rightarrow\infty}G_1(t_k)= 0$ and $\lim\limits_{t_k\rightarrow\infty}G_2(t_k)= 0$ exponentially, $\lim\limits_{t_k\rightarrow\infty} x(t_k)= 0$ exponentially for any $x(t_0)\in \mathds{R}^{n}$.
\end{Lemma}
\section{Non-pathological Sampling and The Convergence Analysis}\label{mainresults}
{\myr In light of the pathological sampling, a significant design issue of the sampled-data systems (\cite[Section~3.2]{chen2012optimal})
has rarely been treated in the existing literature related to distributed systems.
Pathological sampling means the discretized version of a stabilizable/detectable continuous-time plant is not itself stabilizable and detectable \citep*{middleton1995non}.
Pathological sampling can destroy the stabilizability and detectability of the system \eqref{compensator} under zero-order hold discretization. Thus, it is necessary to avoid pathological
sampling periods. In the following, we will fully reveal the interactions among sampling periods, topologies,
reference signals, and related observability; and give all non-pathological / pathological
sampling periods.
}
\subsection{Error dynamics of the mixed time- and event-triggered observer}
Define the \emph{estimation error} and the \emph{relative auxiliary error} of follower $i$ as $\tilde{\eta}_i(t)\triangleq\eta_i(t)-v(t)$ and $\hat{\eta}_{ei}(t) \triangleq \sum_{j\in\mathcal{N}_i}{(\hat{\eta}_j(t)-\hat{\eta}_i(t))}$, respectively. Let $\tilde{\eta} \triangleq \col( \tilde{\eta}_1, \cdots, \tilde{\eta}_N )$, $\hat{\eta}_{e} \triangleq \col(\hat{\eta}_{e1}, \cdots, \hat{\eta}_{eN})$ and $\bar{\eta}\triangleq\col(\bar{\eta}_{ 1},\cdots,\bar{\eta}_{N})$ denote the stacked versions of related signals. Then, the estimation error dynamics can be described by
\begin{equation}\label{reeq}
  \dot{\tilde{\eta}}(t) =\left(I_N\otimes S\right)\tilde{\eta}(t) +\mu \hat{\eta}_{e}(t_k),~t\in\left[t_{k}~t_{k+1}\right).
\end{equation}
At each sampling instant, the stacked relative auxiliary error can be represented as
\begin{equation}
\hat{\eta}_{e}(t_k) =-\left(H\otimes I_n\right)\left(\tilde{\eta}(t_k)-\bar{\eta}(t_k)\right).\label{relaH} \end{equation}
Combining (\ref{reeq}) and (\ref{relaH}) yields a compact representation of the estimation error dynamics for $t\in [t_{k},t_{k+1})$ as follows
\begin{equation}
\dot{\tilde{\eta}}(t)=\left(I_N\otimes S\right)\tilde{\eta}(t)-\mu\left(H\otimes I_n\right)\left(\tilde{\eta}(t_k)-\bar{\eta}(t_k)\right). \label{reeq1}
\end{equation}
By applying the step-invariant transformation \citep*{chen2012optimal}, the continuous-time system in \eqref{reeq1} can be discretized to the system
\begin{equation}\label{solution1}
\tilde{\eta}(t_{k+1})=F\left(\mu\right)\tilde{\eta}(t_k)+G\left(\mu\right)\bar{\eta}(t_k),
\end{equation}
where
\begin{align*}
F\left(\mu\right)&=I_N\otimes e^{Sh}-\mu H\otimes \int_{0}^{h}e^{S\tau}d\tau,\\
G\left(\mu\right)&=\mu H\otimes \int_{0}^{h}e^{S\tau}d\tau.
\end{align*}
If events are triggered at all sampling instants, that is, $s^i_{l}\equiv1$, then $\bar{\eta}(t_k)\equiv0$, and (\ref{solution1}) reduces to the following purely time-triggered error dynamic system
\begin{equation}\label{solutionsa1}
\tilde{\eta}(t_{k+1})=F\left(\mu\right)\tilde{\eta}(t_k).
\end{equation}
{\myr
Then, the matrix $F\left(\mu\right)$ satisfies the following property.
\begin{Lemma}\label{}
All the eigenvalues of $F\left(\mu\right)$ are given by:
\begin{equation}\label{eiFmu}
e^{\lambda_q h}-\mu\lambda_i\int_{0}^{h} e^{\lambda_q\tau}d\tau,~~q\in\mathcal{Q}~~\textnormal{and}~~i\in\mathcal{V}_F.
\end{equation}
\end{Lemma}
\begin{proof}
Please see Appendix B
\end{proof}

}
\subsection{Non-pathological sampling periods and convergence of the time-triggered error dynamics}
To drive the estimation error, $\tilde{\eta}(t_{k})$ in (\ref{solution1}) to zero as time tends to infinity, the stability of the time-triggered error dynamics in (\ref{solutionsa1}) is critical. The following lemma gives a necessary and sufficient condition for the stability of (\ref{solutionsa1}) and reveals some fundamental relationships among sampling periods, topologies, and reference signals.
\begin{Theorem}\label{radio} Under Assumption \ref{ass0}, there exists a $\mu>0$ such that $F\left(\mu\right)$ is Schur if and only if $h$ satisfies the conditions below:
\begin{enumerate}[(I)]
\item \label{codn1} Phase condition:
   $$\psi_{i,q}(h)\in(-\frac{\pi}{2}, \frac{\pi}{2}),~~~~q\in \mathcal{Q}_1\cup Q_3~~\textnormal{and}~~i\in\mathcal{V}_F;$$
\item \label{codn2}  Magnitude condition: $$e^{2\re(\lambda_q)h}<\csc^2\psi_{i,q}(h),~~~~q\in \mathcal{Q}_1~~\textnormal{and}~~i\in\mathcal{V}_F;$$
\item \label{codn2+} Non-zero spectral mapping: $$h\neq \frac{2\kappa\pi}{\im(\lambda_q)},~~~~q\in \mathcal{Q}_3~~\textnormal{and}~~i\in\mathcal{V}_F,$$ where $\kappa \in \mathds{Z}$ and $\frac{\kappa}{\im(\lambda_q)}>0;$
\item \label{codn3} $\bigcap\limits_{q\in \mathcal{Q},~i\in \mathcal{V}_F}(\underline{\mu}_{i,q},\bar{\mu}_{i,q})\neq \emptyset$,
\end{enumerate}
where
\begin{align}
\bar{\mu}_{i,q}=&{\textstyle \frac{|\lambda_q|}{|\lambda_i|}\frac{\cos(\psi_{i,q}(h))+\sqrt{e^{-2\re(\lambda_q)h}-\sin^2(\psi_{i,q}(h))}}{e^{-\re(\lambda_q)h}\sqrt{V_{q}^2(h)+U_{q}^2(h)}}},\nonumber\\
\underline{\mu}_{i,q}=&{\textstyle \frac{|\lambda_q|}{|\lambda_i|}\frac{\cos(\psi_{i,q}(h))-\sqrt{e^{-2\re(\lambda_q)h}-\sin^2(\psi_{i,q}(h))}}{e^{-\re(\lambda_q)h}\sqrt{V_{q}^2(h)+U_{q}^2(h)}}}.\label{solutionmu}
\end{align}
\end{Theorem}
\begin{proof}
It can be verified from the property of Kronecker product and the spectral mapping theorem that the eigenvalues of $F\left(\mu\right)$ are given by:
\begin{equation}
e^{\lambda_q h}-\mu\lambda_i\int_{0}^{h} e^{\lambda_q\tau}d\tau,~~q\in\mathcal{Q}~~\textnormal{and}~~i\in\mathcal{V}_F.\nonumber
\end{equation}
Obviously, $F\left(\mu\right)$ is Schur if and only if these eigenvalues are in the unit circle. Direct calculation gives that
\begin{equation}\label{rhomu}
|e^{\lambda_qh}-\mu\lambda_i\int_{0}^{h} e^{\lambda_q\tau}d\tau\big|^2-1=\alpha_{i,q} \mu^2+\beta_{i,q}\mu+\gamma_{i,q},
\end{equation}
where
\begin{align}\label{albetaga}
  \alpha_{i,q}&= |\lambda_i|^2\big|\int_{0}^{h} e^{\lambda_q\tau}d\tau\big|^2,\nonumber\\
 \beta_{i,q}&=-2\operatorname{Re}(\lambda_i\int_{0}^{h} e^{\lambda_q\tau+\lambda_q^*h}d\tau),\\
 \gamma_{i,q} & =\big|e^{\lambda_q h}\big|^2-1,~~q\in\mathcal{Q}~~\textnormal{and}~~i\in\mathcal{V}_F.\nonumber
\end{align}
Hence, $F\left(\mu\right)$ is Schur if and only if there exists a $\mu>0$ such that $\alpha_{i,q} \mu^2+\beta_{i,q}\mu+\gamma_{i,q}<0$ for all possible $i$ and $q$.
For any $q\in\mathcal{Q}_1 \cup \mathcal{Q}_2 \cup \mathcal{Q}_3 $ and $i\in\mathcal{V}_F$, we have that (see Appendix C for the derivations)
\begin{align}\label{alphaq3q2q4}\textstyle
\alpha_{i,q}
&=\frac{|\lambda_i|^2}{|\lambda_q|^2}\big(U_{q}^2(h)+V_{q}^2(h)\big),\nonumber\\
\beta_{i,q}&=-2\frac{|\lambda_i|}{|\lambda_q|}e^{\re(\lambda_q)h}\sqrt{V_{q}^2(h)+U_{q}^2(h)}\cos(\psi_{i,q}(h)),\\
\gamma_{i,q}&=e^{2\re(\lambda_q)h}-1.\nonumber
\end{align}
For $q\in \mathcal{Q}_4$ and $i\in\mathcal{V}_F$, $\lambda_q=0$, which implies that $\alpha_{i,q}=|\lambda_i|^2h^2$, $\beta_{i,q}=-2 h \re(\lambda_i)$ and $\gamma_{i,q}=0$ from \eqref{albetaga}. It can be evaluated from Lemma \ref{lemmauvpt} that
 \begin{align}\textstyle
\alpha_{i,q}=&\lim_{\lambda_q\rightarrow 0}\frac{|\lambda_i|^2}{|\lambda_q|^2}\big(U_{q}^2(h)+V_{q}^2(h)\big)=|\lambda_i|^2h^2,\nonumber\\
\beta_{i,q}=&-2|\lambda_i|\lim_{\lambda_q\rightarrow 0}{\textstyle \frac{\sqrt{V_{q}^2(h)+U_{q}^2(h)}}{|\lambda_q|}}e^{\re(\lambda_q)h}\cos(\psi_{i,q}(h))\nonumber\\
=&-2 h \re(\lambda_i),\nonumber\\
\gamma_{i,q}=&\lim_{\lambda_q\rightarrow 0}e^{2\re(\lambda_q)h}-1=0.\nonumber
\end{align}
Therefore, for any $q\in\mathcal{Q}$ and $i\in\mathcal{V}_F$, $\alpha_{i,q}$, $\beta_{i,q}$ and $\gamma_{i,q}$ can be expressed by \eqref{alphaq3q2q4}. Based on the partition in (\ref{partitionof1ton}) and the coefficients in \eqref{alphaq3q2q4}, the existence of $\mu>0$ is discussed below.

\textbf{Case 1. $q\in\mathcal{Q}_1$ and $i\in\mathcal{V}_F$:}\\
As $\operatorname{Re}\left(\lambda_q\right)>0$, $\big|e^{\lambda_q h}\big|=\big|e^{\operatorname{Re}\left(\lambda_q\right) h}\big|>1$, which implies $\gamma_{i,q}>0$. Noting that $U_{q}(h)$ in \eqref{UVPT} cannot be zero for any $q\in\mathcal{Q}_1$, we have that $\alpha_{i,q}>0$. Since $\gamma_{i,q}>0$ and $\alpha_{i,q}>0$, there exists a $\mu>0$ such that $\alpha_{i,q} \mu^2+\beta_{i,q}\mu+\gamma_{i,q}<0$ if and only if $\beta_{i,q}<0$ and $ \beta_{i,q}^2-4\alpha_{i,q}\gamma_{i,q}> 0$, which is equivalent to $\cos(\psi_{i,q}(h))>0$ and $1-e^{2\re(\lambda_q)}\sin^2(\psi_{i,q}(h))> 0$ from the derivations in Appendix C. Hence, there exists a $\mu>0$ such that $\alpha_{i,q} \mu^2+\beta_{i,q}\mu+\gamma_{i,q}<0$ if and only if the following two conditions are satisfied
 \begin{enumerate}[(a)]
   \item $2\kappa\pi-\frac{\pi}{2}<\psi_{i,q}(h) <2\kappa\pi+\frac{\pi}{2},~~\kappa\in \mathds{Z}$,
   \item $e^{2\re(\lambda_q)h}<\csc^2\psi_{i,q}(h)$.
 \end{enumerate}
 Owing to $\operatorname{Re}\left(\lambda_q\right)>0$, $\theta_q \in(-\frac{1}{2}\pi,\frac{1}{2}\pi)$ and $U_{q}(h)>0$, it follows that $\phi_{q}(h)\in(-\frac{1}{2}\pi,\frac{1}{2}\pi)$ from \eqref{UVPT}. Under Assumption \ref{ass0}, all eigenvalues of $H$ have positive real parts, which leads to $\theta_i\in (-\frac{1}{2}\pi,\frac{1}{2}\pi)$. As a result, $\psi_{i,q}(h)\in (-\frac{3}{2}\pi, \frac{3}{2}\pi)$. Accordingly, condition (a) reduces to $-\frac{\pi}{2}<\psi_{i,q}(h) <\frac{\pi}{2}$.

\textbf{Case 2. $q\in\mathcal{Q}_2$ and $i\in\mathcal{V}_F$:}\\ As $\operatorname{Re}\left(\lambda_q\right)<0$, $\big|e^{\lambda_q h}\big|=\big|e^{\operatorname{Re}\left(\lambda_q\right) h}\big|<1$, which implies that $\gamma_{i,q}<0$. It follows from \eqref{UVPT} that $U_{q}(h)$ and $V_{q}(h)$ cannot be equal to zero when
$q\in\mathcal{Q}_2$, and thus $\alpha_{i,q}>0$. Since $\alpha_{i,q}>0$ and $\gamma_{i,q}<0$, there exists a $\mu>0$ such that $\alpha_{i,q} \mu^2+\beta_{i,q}\mu+\gamma_{i,q}<0$ if and only if $ \beta_{i,q}^2-4\alpha_{i,q}\gamma_{i,q}> 0$, which is equivalent to $e^{2\re(\lambda_q)h}<\csc^2\psi_{i,q}(h)$ from the derivations in Appendix C. It is noted that $$\sin^2(\psi_{i,q}(h))<1<e^{-2\re(\lambda_q)h},~~\textnormal{for~~any}~~q \in \mathcal{Q}_2.$$ Hence, the condition $e^{2\re(\lambda_q)h}<\csc^2\psi_{i,q}(h)$ always holds, and there exists a $\mu>0$ such that $\alpha_{i,q} \mu^2+\beta_{i,q}\mu+\gamma_{i,q}<0$ for any $q \in \mathcal{Q}_2$ and $h>0$.

\textbf{Case 3. $q\in\mathcal{Q}_3$ and $i\in\mathcal{V}_F$:}\\
As $\operatorname{Re}\left(\lambda_q\right)=0$, $\big|e^{\lambda_q h}\big|=\big|e^{\operatorname{Re}\left(\lambda_q\right) h}\big|=1$, which yields that $\gamma_{i,q}=0$. There exists a $\mu>0$ such that $\alpha_{i,q} \mu^2+\beta_{i,q}\mu <0$ if and only if $\beta_{i,q}<0$, which is equivalent to $\cos(\psi_{i,q}(h))>0$ and $\sqrt{V_{q}^2(h)+U_{q}^2(h)}\neq0$ from (\ref{alphaq3q2q4}). Hence, there exists a $\mu>0$ such that $\alpha_{i,q} \mu^2+\beta_{i,q}\mu <0$ if and only if $h$ satisfies
\begin{enumerate}[(a)]
   \item $2\kappa\pi-\frac{\pi}{2}<\psi_{i,q}(h) <2\kappa\pi+\frac{\pi}{2},~~\kappa\in \mathds{Z}$,
   \item $\im(\lambda_q)h\neq 2\kappa\pi,~~\kappa\in \mathds{Z}$.
 \end{enumerate}
When $\im(\lambda_q)h= 2\kappa\pi$, $U_{q}(h)=V_{q}(h)=0$, and $\phi_{q}(h)$ can be determined by using the limiting process, that is, $\phi_{q}(h)=\pm \frac{\pi}{2}$. Otherwise, $U_{q}(h)>0$, which implies that $\phi_{q}(h)\in(-\frac{1}{2}\pi,\frac{1}{2}\pi)$. Hence, $\phi_{q}(h)\in[-\frac{1}{2}\pi,\frac{1}{2}\pi]$. Noting that $\theta_i\in (-\frac{1}{2}\pi,\frac{1}{2}\pi)$ and $\theta_q =\pm\frac{\pi}{2}$, one has that $\psi_{i,q}(h)\in (-\frac{3}{2}\pi, \frac{3}{2}\pi)$, and, therefore, condition (a) reduces to $\label{psiii}-\frac{\pi}{2}<\psi_{i,q}(h) <\frac{\pi}{2}$. Regarding condition (b), it reduces to $h\neq \frac{2\kappa\pi}{\im(\lambda_q)}$, $\frac{\kappa}{\im(\lambda_q)}>0$, since $h$ is positive.

\textbf{Case 4. $q\in\mathcal{Q}_4$ and $i\in\mathcal{V}_F$:}\\
For $q\in \mathcal{Q}_4$, $\lambda_q=0$ implies $\alpha_{i,q}=|\lambda_i|^2h^2$, $\beta_{i,q}=-2 h \re(\lambda_i)$ and $\gamma_{i,q}=0$ from \eqref{albetaga}.
Under Assumption \ref{ass0}, all eigenvalues of $H$ have positive real parts and $h>0$, which leads to $\beta_{i,q}<0$ and $\alpha_{i,q}>0$. Thus, $\alpha_{i,q} \mu^2+\beta_{i,q}\mu<0$ always has positive solutions for any $q\in \mathcal{Q}_4$ and $h>0$.

Next, we show that for each $q\in \mathcal{Q}$, the solutions of $\alpha_{i,q} \mu^2+\beta_{i,q}\mu+\gamma_{i,q}=0$ can be represented by \eqref{solutionmu}. For any $q\in \mathcal{Q}_1\cup\mathcal{Q}_2\cup\mathcal{Q}_3$ and $i\in \mathcal{V}_F$, it is easy to verify that the roots of $\alpha_{i,q} \mu^2+\beta_{i,q}\mu+\gamma_{i,q}=0$ are $\bar{\mu}_{i,q}$ and $\underline{\mu}_{i,q}$ in \eqref{solutionmu}. Under Assumption \ref{ass0}, all eigenvalues of $H$ have positive real parts, that is, $\cos(\theta_i)=\frac{\re(\lambda_i)}{|\lambda_i|}>0$. By using Lemma \ref{lemmauvpt}, for the $\bar{\mu}_{i,q}$ and $\underline{\mu}_{i,q}$ defined in \eqref{solutionmu}, we have that
\begin{align}\label{soultion00limited}
\lim_{\lambda_q\rightarrow0}\bar{\mu}_{i,q}
=&\frac{1}{|\lambda_i|}\frac{\cos(\theta_{i})+|\cos(\theta_{i})|}{h}
=\frac{2\re(\lambda_i)}{|\lambda_i|^2h},\nonumber\\
\lim_{\lambda_q\rightarrow0}\underline{\mu}_{i,q}
=&\frac{1}{|\lambda_i|}\frac{\cos(\theta_{i})-|\cos(\theta_{i})|}{h}=0.
\end{align}
In addition, for $q\in \mathcal{Q}_4$ and $i\in \mathcal{V}_F$, the roots of $\alpha_{i,q} \mu^2+\beta_{i,q}\mu+\gamma_{i,q}=0$ are \begin{align}\label{soultion00}\underline{\mu}_{i,q}=0~~\textnormal{and} ~~\bar{\mu}_{i,q}=\frac{2\re(\lambda_i)}{|\lambda_i|^2h}, \end{align}
respectively.
Hence, the equations in \eqref{solutionmu} are equivalent to the equations in \eqref{soultion00} when $q\in \mathcal{Q}_4$ because of \eqref{soultion00limited}.

In conclusion, $\mu>0$ exists if and only if conditions \eqref{codn1} to \eqref{codn3} are all satisfied.
\end{proof}
\begin{Remark}
Theorem \ref{radio} gives all feasible sampling periods that lead to convergent error dynamics and observable connected systems. This can be regarded as an extension of the non-pathological sampling periods for observability in \cite{chen2012optimal} to distributed systems.
\end{Remark}

\subsection{The interplay between phrase of the graph and non-pathological sampling periods}
A significant feature of the distributed system considered in this paper is that the observability under-sampling is related not only to the system dynamics but also to the topological connection. According to Lemma \ref{lemmauvpt} and \eqref{UVPT}, $$\lim\limits_{h\rightarrow0}\frac{U_{q}\left(h\right)+jV_{q}\left(h\right)}{h\lambda_q}=1.$$ In terms of this, $\frac{U_{q}\left(h\right)+jV_{q}\left(h\right)}{h}$ can be regarded as the sampled version of $\lambda_q$ with the proposed distributed observers, and $\phi_q\left(h\right)-\theta_q$ is the phase shift of $\lambda_q$ due to sampling. These phase shifts of the eigenvalues of $S$ and the phases $\theta_i$ related to the graph jointly determine the observability as shown in condition \eqref{codn1}.

{\myr
In the line of the Magnitude condition \eqref{codn2} of Theorem \ref{radio},
the terms $e^{2\re(\lambda_q)h}$ and $\csc^2\psi_{i,q}(h)$ represent the divergence rate of the observed system and the tracking performance of the observer network, respectively. We now give two exceptional cases to illustrate the interplay between the phrase of the graph and non-pathological sampling periods:
}

{\myr \textbf{Special Case 1:} $\theta_i\equiv\theta_q$, $\forall i\in \mathcal{V}_F$ and $\forall q\in \mathds{Q}_1\equiv \mathds{Q}$($N=n$). We now show $h$ can be arbitrarily large for the Phase condition and Magnitude condition of Theorem \ref{radio} except for some isolated points:
\begin{enumerate}
  \item[1:] Owing to $\operatorname{Re}\left(\lambda_q\right)>0$, $\theta_q \in(-\frac{1}{2}\pi,\frac{1}{2}\pi)$ and $U_{q}(h)>0$, it follows that $\phi_{q}(h)\in(-\frac{1}{2}\pi,\frac{1}{2}\pi)$ from \eqref{UVPT}. It can be easy to conclude that the phase condition of Theorem \ref{radio} always holds because of $\psi_{i,q}(h)=\phi_{q}(h)$.
  \item[2:] Magnitude condition is equivalent to $$e^{2\re(\lambda_q)h}\sin^2\phi_{q}(h)<1.$$
 Hence, from equation \eqref{UVPT},  the Magnitude condition is satisfied if only if
  \begin{align*}
  e^{2\re(\lambda_q)h} V_{q}^2(h)<&U_{q}^2(h)+V_{q}^2(h).
    \end{align*}
    Hence, we have
      \begin{align}\label{magninegative}
     &e^{2\re(\lambda_q)h} V_{q}^2(h)-U_{q}^2(h)-V_{q}^2(h)\nonumber\\
 = & e^{2\re(\lambda_q)h}\sin^2(\im(\lambda_q)h)-e^{2\re(\lambda_q)h}\nonumber\\
 &+2e^{\re(\lambda_q)h}\cos(\im(\lambda_q)h)-1\nonumber\\
 = &- e^{2\re(\lambda_q)h}\cos^2(\im(\lambda_q)h)\nonumber\\
 &+2e^{\re(\lambda_q)h}\cos(\im(\lambda_q)h)-1\nonumber\\
 = &-\big[e^{\re(\lambda_q)h}\cos(\im(\lambda_q)h)-1\big]^2.
  \end{align}
Thus, from \eqref{magninegative}, Magnitude condition is satisfied for any arbitrarily large $h$ except for some isolated points $\{h|h\in \mathds{R}^{+},e^{\re(\lambda_q)h}\cos(\im(\lambda_q)h)=1\}$.
\end{enumerate}
}
\begin{Remark}\label{Speremark3}\myr An interesting result of the  Special Case 1 is that $\theta_i\equiv\theta_q\equiv 0$. The solution to $\alpha_{i,q} \mu^2+\beta_{i,q}\mu+\gamma_{i,q}<0$ for all possible $i$ and $q$ is
\begin{align}
\bar{\mu}_{i,q}={\textstyle \frac{|\lambda_q|}{|\lambda_i|}\frac{1+e^{-\re(\lambda_q)h}}{1-e^{-\re(\lambda_q)h}}},
~~\textnormal{and}~~\underline{\mu}_{i,q}={\textstyle \frac{|\lambda_q|}{|\lambda_i|}}.\label{solutionmuspecial}
\end{align}
If for all possible $i$ and $q$, $\frac{|\lambda_q|}{|\lambda_i|}$ are equal, then, the intersection of the solutions \eqref{solutionmuspecial} is $\big(\frac{|\lambda_q|}{|\lambda_i|}, \frac{|\lambda_q|}{|\lambda_i|}\frac{1+e^{-\re(\lambda_q)h}}{1-e^{-\re(\lambda_q)h}}\big)$ under this condition. Thus, the sampling period $h$ can be arbitrarily large. However, the Lebesgue measure of possible ranges $\mu$ will decrease exponentially as $h$ tends to infinity.
\end{Remark}

{\myr \textbf{Special Case 2:} $\theta_i\equiv0 $ and $\theta_q\neq0 $, $\forall i\in \mathcal{V}_F$ and $\forall q\in \mathds{Q}_1\equiv \mathds{Q}$($N=n$). We now show $h$ can't be arbitrarily large for the Phase condition and Magnitude condition of Theorem \ref{radio}:
\begin{enumerate}
  \item[1:] From equation \eqref{UVPT}, we have
  \begin{align*}
  \cos\big(\phi_{q}(h)&-\theta_q\big)=\cos\phi_{q}(h)\cos\theta_q+\sin\phi_{q}(h)\sin\theta_q\\
  =&{\textstyle \frac{U_{q}(h)\re(\lambda_q)+V_{q}(h)\im(\lambda_q)}{|\lambda_q|\sqrt{V_{q}^2(h)+U_{q}^2(h)}}}\\
  =&{\textstyle \frac{(e^{\re(\lambda_q)h}-\cos(\im(\lambda_q)h))\re(\lambda_q)+\sin(\im(\lambda_q)h)\im(\lambda_q)}{|\lambda_q|\sqrt{V_{q}^2(h)+U_{q}^2(h)}}}.
  \end{align*}
  Thus, the phase condition may be violated. For example,
 $\re(\lambda_q)=10^{-5}$, $\im(\lambda_q)=\pi$ and $h=3/2$ lead to
   $\cos\big(\phi_{q}(h)-\theta_q\big)<0$, which implies phase condition is not satisfied. However, when $h\rightarrow0$, $\sin(\im(\lambda_q)h)\im(\lambda_q)$ will be positive which implies $\cos\big(\phi_{q}(h)-\theta_q\big)>0$. Thus, this implies phase condition is always satisfied for small enough $h$.
  \item[2:] Magnitude condition is equivalent to $$e^{2\re(\lambda_q)h}\sin^2\big(\phi_{q}(h)-\theta_q\big)<1.$$
      From Lemma \ref{lemmauvpt}, we have $\lim_{h\rightarrow0}e^{2\re(\lambda_q)h}\sin^2\big(\phi_{q}(h)-\theta_q\big)=0 $
      and $\lim_{h\rightarrow\infty}e^{2\re(\lambda_q)h}\sin^2\big(\phi_{q}(h)-\theta_q\big)=\infty $.
      Hence,  the Magnitude condition is satisfied for some small $h$.
\end{enumerate}
}

{\myr To summarise: Special Case 1: $h$ can be arbitrarily large. Special Case 2: $h$ cannot be arbitrarily large. The phases of the graph have vital interactions with non-pathological sampling and related observer design. To some extent, Special case 1 can be regarded as "phase alignment" between the system observed and the observer network. Besides, to guarantee the observability, the sampled eigenvalues $\frac{U_{q}\left(h\right)+jV_{q}\left(h\right)}{h}$ as revealed in condition \eqref{codn2+} must be non-zero, which corresponds to the temporal anti-aliasing of signals, consistent with the centralized control and signal processing.
}

\subsection{The convergence rate of the time-triggered distributed observers}
{\myr For the system \eqref{solutionsa1}, the convergence rate is determined by the eigenvalues of $F(\mu)$. Thus, it will be interesting to investigate the properties of the Spectral radius related to $F(\mu)$.}

{\myr
\begin{Property} \label{proper1}$\min\limits_{\mu>0} \rho\big(F(\mu)\big)\geq\max\limits_{i\in \mathcal{V}_F,~q\in \mathds{Q}}\frac{e^{\re(\lambda_q)h}}{|\csc(\psi_{i,q}(h))|}$.
\end{Property}
\begin{proof} For any $h$ satisfies the conditions in Theorem \ref{radio}, from \eqref{rhomu},
$$\rho^2\big(F(\mu)\big)\in\{\alpha_{i,q} \mu^2+\beta_{i,q}\mu+\gamma_{i,q}+1, i\in \mathcal{V}_F,q \in \mathds{Q}\}.$$
 $\alpha_{i,q}$ is positive for any $i\in \mathcal{V}_F$ and $q\in \mathds{Q}$. Hence, the minimal value of $\alpha_{i,q} \mu^2+\beta_{i,q}\mu+\gamma_{i,q}+1$ is $-\frac{\beta_{i,q}^2-4\alpha_{i,q}\gamma_{i,q}}{4\alpha_{i,q}}+1$ at $\mu=-\frac{\beta_{i,q}}{2\alpha_{i,q}}$. Thus,
\begin{align*}
\min\limits_{\mu>0} \rho^2\big(F(\mu)\big)\geq&-\frac{\beta_{i,q}^2}{4\alpha_{i,q}}+\gamma_{i,q}+1\\
=& e^{2\re(\lambda_q)h}-\cos^2(\psi_{i,q}(h))e^{2\re(\lambda_q)h}\\
=& e^{2\re(\lambda_q)h}\sin^2(\psi_{i,q}(h)).
\end{align*}
Hence, the Spectral radius of $F(\mu)$ satisfies Property \ref{proper1}.
\end{proof}
}
{\myr Property \ref{proper1} reveals that the fastest convergence rate is related to
\begin{enumerate}[A]
  \item The divergence/convergence rate of the observed system ($e^{\re(\lambda_q h)}$).
  \item The tracking rate of the observer network ($\csc(\psi_{i,q}(h))$).
\end{enumerate}
Normally, the lower bound is not achievable with the proposed observer, and a more general architecture is needed. It will be interesting to find the fastest convergence rate, i.e. $\min\limits_{\mu>0} \rho\big(F(\mu)\big)$ for any non-pathological sampling $h$. We suppose all the eigenvalues of $S$ have non-negative real parts. Then, the Spectral radius related to $F(\mu)$ satisfies the following property.
}

{\myr
\begin{Property} \label{proper2}
The optimal solution to $\min\limits_{\mu > 0}\rho(F(\mu))$ is located at some intersection points of the parabolic curves $y_{i,q}(\mu)=\alpha_{i,q}\mu^2+\beta_{i,q}\mu+\gamma_{i,q}$ for $i\in \mathds{N}$ and $q\in \mathds{Q}$.
\end{Property}}
\begin{proof}\myr For any $h$ satisfies the conditions in Theorem \ref{radio}, $\bigcap\limits_{q\in \mathcal{Q},~i\in \mathcal{V}_F}(\underline{\mu}_{i,q},\bar{\mu}_{i,q})$ is non-empty. Clearly, The optimal solution to $\min\limits_{\mu > 0}\rho^2(F(\mu))$ exists and is in the following closed interval $\underline{\mu}\leq \mu \leq \bar{\mu}$ where $\underline{\mu}=\max\limits_{i\in \mathcal{V}_F, q\in \mathds{Q}}\underline{\mu}_{i,q}$ and $\bar{\mu}=\min\limits_{i\in \mathcal{V}_F, q\in \mathds{Q}}\bar{\mu}_{i,q}$. Since all the eigenvalues of $S$ have non-negative real parts, $\underline{\mu}_{i,q}\geq 0$ for any $i\in \mathcal{V}_F$ and $q\in \mathds{Q}$ from the proof of Theorem \ref{radio}(Case 1, Case 3 and Case 4). Thus, $\underline{\mu}$ and $\bar{\mu}$ are the solution to $y_{i,q}(\mu)=0$, which implies $\rho^2(F(\mu)\big|_{\mu=\bar{\mu}~~or~~\bar{\mu}}=1$. For any $\mu$ in $\bigcap\limits_{q\in \mathcal{Q},~i\in \mathcal{V}_F}(\underline{\mu}_{i,q},\bar{\mu}_{i,q})$, we have $y_{i,q}(\mu)<0$. Hence, the optimal solution to $\min\limits_{\mu > 0}\rho^2(F(\mu))=\max\limits_{q\in \mathcal{Q},~i\in \mathcal{V}_F}\{y_{i,q}(\mu)+1\}$ is in the $(\underline{\mu}, \bar{\mu})=\bigcap\limits_{q\in \mathcal{Q},~i\in \mathcal{V}_F}(\underline{\mu}_{i,q},\bar{\mu}_{i,q})$.

\end{proof}

\subsection{Relationships between time-triggered distributed observers and continuous case}
The following corollary shows that Theorem \ref{radio} includes the result in  \cite{su2011cooperative} as a special case.
\begin{Corollary}\label{lemmaeigen}For $i\in \mathcal{V}_F$ and $q\in \mathcal{Q}$, $\lim\limits_{h\rightarrow 0}\bar{\mu}_{i,q}= +\infty$, $\lim\limits_{h\rightarrow 0}\underline{\mu}_{i,q}=\frac{\re(\lambda_q)}{\re(\lambda_i)}$, and
\begin{equation*}
 {\textstyle \lim\limits_{h\rightarrow 0} \bigcap\limits_{q\in \mathcal{Q},~i\in \mathcal{V}_F}(\underline{\mu}_{i,q},\bar{\mu}_{i,q})=\big(\frac{\max_{q\in \mathcal{Q}}(\re(\lambda_q))}{\min_{i\in \mathcal{V}_F}(\re(\lambda_i))},+\infty\big)}
\end{equation*}
where $\bar{\mu}_{i,q}$ and $\underline{\mu}_{i,q}$ are defined in \eqref{solutionmu}.
\end{Corollary}
\begin{proof}
It follows from Lemma \ref{lemmauvpt} and Appendix C that
\begin{subequations}\label{abc}
\begin{align}
\lim_{h\rightarrow 0}\frac{\alpha_{i,q}}{h}=&\lim_{h\rightarrow 0}\frac{|\lambda_i|^2}{|\lambda_q|^2}\frac{V_{q}^2(h)+U_{q}^2(h)}{h}=0,\\
\lim_{h\rightarrow 0}\frac{\beta_{i,q}}{h}=&-2\re(\lambda_i),\\
\lim_{h\rightarrow 0}\frac{\gamma_{i,q}}{h}=&2\re(\lambda_q).
\end{align}
\end{subequations}
where $\alpha_{i,q}$, $\beta_{i,q}$, and $\gamma_{i,q}$ are defined in (\ref{albetaga}).
For $i\in \mathcal{V}_F$ and $q\in \mathcal{Q}$, the solutions to $\alpha_{i,q} \mu^2+\beta_{i,q}\mu+\gamma_{i,q}=0$ are given by $\bar{\mu}_{i,q}=\frac{-\beta_{i,q}+\sqrt{\beta_{i,q}^2-4\alpha_{i,q}\gamma_{i,q}}}{2\alpha_{i,q}}$ and $\underline{\mu}_{i,q}=\frac{-\beta_{i,q}-\sqrt{\beta_{i,q}^2-4\alpha_{i,q}\gamma_{i,q}}}{2\alpha_{i,q}}$, respectively. Under Assumption \ref{ass0}, all eigenvalues of $H$ have positive real parts, which leads to ${\textstyle\lim\limits_{h\rightarrow 0}\sqrt{\frac{\beta_{i,q}^2}{h^2}-4\frac{\alpha_{i,q}\gamma_{i,q}}{h^2}}=2\re(\lambda_i)}$. Based on this and (\ref{abc}), we have that
\begin{align*}
\lim\nolimits_{h\rightarrow 0}\bar{\mu}_{i,q}=&{\textstyle \lim_{h\rightarrow 0}\frac{-\beta_{i,q}/h+\sqrt{\beta_{i,q}^2-4\alpha_{i,q}\gamma_{i,q}}/h}{2\alpha_{i,q}/h}}\\
=&{\textstyle\frac{4\re(\lambda_i)}{0}}\\
=&+\infty
\end{align*}
and
\begin{align*}
  \lim\nolimits_{h\rightarrow 0}\underline{\mu}_{i,q}= & {\textstyle\lim_{h\rightarrow 0}\frac{2\gamma_{i,q}}{-\beta_{i,q}+\sqrt{\beta_{i,q}^2-4\alpha_{i,q}\gamma_{i,q}}}} \\
  =& {\textstyle\lim_{h\rightarrow 0}\frac{2\gamma_{i,q}/h}{-\beta_{i,q}/h+\sqrt{\beta_{i,q}^2-4\alpha_{i,q}\gamma_{i,q}}/h}} \\
  = & {\textstyle\frac{\re(\lambda_q)}{\re(\lambda_i)}}.
\end{align*}
Finally, $ \lim\limits_{h\rightarrow 0} \bigcap\limits_{q\in \mathcal{Q},~i\in \mathcal{V}_F}(\underline{\mu}_{i,q},\bar{\mu}_{i,q})=(\frac{\max_{q\in \mathcal{Q}}(\re(\lambda_q))}{\min_{i\in \mathcal{V}_F}(\re(\lambda_i))},+\infty)$. This completes the proof.
\end{proof}

On the one hand, the observer gain $\mu$ for the continuous-time distributed observer proposed in \cite{su2011cooperative} should be selected as an arbitrarily large value such that $\mu\in \big(\frac{\max_{q\in \mathcal{Q}}(\re(\lambda_q))}{\min_{i\in \mathcal{V}_F}(\re(\lambda_i))},+\infty\big)$ (See the proof of \cite[Theorem~1]{su2011cooperative}). Clearly, Corollary \ref{lemmaeigen} directly shows that Theorem \ref{radio} includes the result in \cite{su2011cooperative} as a special case when $h\rightarrow 0$. Furthermore, in our time-triggered distributed observer, $\mu$ cannot be selected arbitrarily, as revealed in Theorem \ref{radio}, due to the interaction among sampling periods, topologies, and system dynamics. Instead, to generate convergent error dynamics, $\mu$ needs to be designed deliberately by considering the interaction.

On the other hand, in the existing literature, the sampling period $h$ is usually selected as a sufficiently small number or restricted to some values on a fixed interval governed by linear matrix inequalities or the spectral radius/norm of the certain graph-induced matrix, e.g., \cite{cao2015event} and \cite{zheng2020periodic}. By contrast, a complete characterization of all possible sampling periods and their impact on the observability of the connected systems have been given in Theorem \ref{radio}. In the following corollary, explicit feasible sampling periods are further provided when the reference signal is a sinusoidal wave.

\begin{Corollary}\label{corral} Under Assumption \ref{ass0} and $\mathcal{Q}=\mathcal{Q}_3$, there exists a $\mu$ such that $F\left(\mu\right)$ is Schur if and only if $h$ satisfies the following conditions for $q\in \mathcal{Q}$ and $\kappa\in \mathds{Z}^{+}$:
     \begin{align} \label{h4Q3}\left\{
  \begin{array}{c}
    2\kappa\pi-2\pi<|\im(\lambda_q)| h<2\kappa\pi ,\\
2\kappa\pi-3\pi+2\theta_i<|\im(\lambda_q)|h<2\kappa\pi-\pi+2\theta_i,\\
2\kappa\pi-3\pi-2\theta_i<|\im(\lambda_q)|h<2\kappa\pi-\pi-2\theta_i.\\
     \end{array}
\right.\end{align}
  \end{Corollary}
\begin{proof}
 For $q\in \mathcal{Q}=\mathcal{Q}_3$, $\operatorname{Re}(\lambda_q)=0$ implies $\theta_q=\pm\frac{\pi}{2}$, $U_q(h)=1-\cos(\im(\lambda_q)h)$, and $V_q(h)=\sin(\im(\lambda_q)h)$. Then, simple manipulations give that conditions (\ref{codn1})-(\ref{codn2+}) in Theorem \ref{radio} are equivalent to
 \begin{align}\label{intervalphithe}
 \left\{\begin{array}{c}
 |\im(\lambda_q)|h\neq 2\kappa_1 \pi,~~~~\kappa_1\in \mathds{Z}^{+}\\
 0< \phi_{q}(h)+\theta_i< \pi,~~~~ \text{when}~~\theta_q=\frac{\pi}{2},\\
 -\pi< \phi_{q}(h)+\theta_i< 0,~~~~\text{when}~~ \theta_q=-\frac{\pi}{2}.
 \end{array}\right.
 \end{align}
Since $U_{q}>0$ under these conditions, we have that $\phi_{q}(h)\in (-\frac{\pi}{2},\frac{\pi}{2})$ and
\begin{align}
\phi_{q}(h)&=\arctan\frac{V_{q}(h)}{U_{q}(h)}\nonumber\\
=&\arctan\frac{\sin(\im(\lambda_q)h)}{1-\cos(\im(\lambda_q)h)}\nonumber\\
=& \arctan \big[\cot(\frac{\im(\lambda_q)}{2}h)\big]\nonumber\\
=& \arctan \big[\tan(\kappa\pi+\frac{\pi}{2}-\frac{\im(\lambda_q)}{2}h)\big], ~~~~\kappa\in \mathds{Z},
\end{align}
which further leads to
\begin{align}
\left\{\begin{array}{c}
-\frac{\pi}{2}<\kappa_2\pi+\frac{\pi}{2}-\frac{|\im(\lambda_q)|}{2}h< \frac{\pi}{2},\\
0<\kappa_2\pi+\frac{\pi}{2}-\frac{|\im(\lambda_q)|}{2}h+\theta_i<\pi,
\end{array}\right. \theta_q=\frac{\pi}{2},~\kappa_2\in \mathds{Z},\nonumber
\end{align}
and
\begin{align}
\left\{\begin{array}{c}
-\frac{\pi}{2}<\kappa_3\pi+\frac{\pi}{2}+\frac{|\im(\lambda_q)|}{2}h< \frac{\pi}{2},\\
-\pi<\kappa_3\pi+\frac{\pi}{2}+\frac{|\im(\lambda_q)|}{2}h+\theta_i<0,
\end{array}\right. \theta_q=-\frac{\pi}{2},~\kappa_3\in \mathds{Z}.\nonumber
\end{align}
Through some simple algebraic manipulations, these two conditions can be further simplified to
(1) for $\theta_q=\frac{\pi}{2}$ and $\kappa_2\in \{0,1,2,\cdots\}$,
\begin{align}
\left\{\begin{array}{c}
2\kappa_2\pi<|\im(\lambda_q)|h< 2\kappa_2\pi+2\pi,\\
2\kappa_2\pi-\pi+2\theta_i<|\im(\lambda_q)|h<2\kappa_2\pi+\pi+2\theta_i,
\end{array}\right.\nonumber
\end{align}
and (2) for $\theta_q=-\frac{\pi}{2}$ and $\kappa_3\in \{-1,-2,\cdots\}$,
\begin{align}
\left\{\begin{array}{c}
-2\kappa_3\pi-2\pi<|\im(\lambda_q)|h<-2\kappa_3\pi,\\
-2\kappa_3\pi-3\pi-2\theta_i<|\im(\lambda_q)|h<-2\kappa_3\pi-\pi-2\theta_i.
\end{array}\right.\nonumber
\end{align}
It is obvious that conditions (1) and (2) naturally include the condition $|\im(\lambda_q)|h\neq 2\kappa_1 \pi,~\kappa_1\in \mathds{Z}^{+}$, and are equivalent to the conditions in (\ref{h4Q3}).

From \eqref{alphaq3q2q4}, we have $\alpha_{i,q}=2\frac{|\lambda_i|^2}{|\lambda_q|^2}\big(1-\cos(\im(\lambda_q)h)\big)$ and $$\beta_{i,q}=\left\{\begin{array}{c}-\frac{|\lambda_i|\sqrt{2-2\cos(\im(\lambda_q)h)}}{|\lambda_q|}\sin( \phi_{q}(h)+\theta_i),~~\theta_q=\frac{\pi}{2},\\
\frac{|\lambda_i|\sqrt{2-2\cos(\im(\lambda_q)h)}}{|\lambda_q|}\sin( \phi_{q}(h)+\theta_i),~~ \theta_q=-\frac{\pi}{2},\\
\end{array}\right.\nonumber$$
which implies $\alpha_{i,q}>0$ and $\beta_{i,q}<0$ for any $q\in \mathcal{Q}=\mathcal{Q}_3$ and $i\in \mathcal{V}_F$. Therefore, condition (\ref{codn3}) in Theorem \ref{radio} always holds with $\underline{\mu}_{i,q}=0$ and $\bar{\mu}_{i,q}=-\frac{\beta_{i,q}}{\alpha_{i,q}}>0$ for any $q\in \mathcal{Q}=\mathcal{Q}_3$ and $i\in \mathcal{V}_F$. This completes the proof.
\end{proof}
{\myb
\subsection{An approach to cancel the effect of imaginary parts}
The imaginary parts of the leader-following network-induced matrix $H$ play an important role in selecting the sampling period and generate more possibility of a pathological sampling period. In order to cancel the effect of imaginary parts induced by the network, we now modify \eqref{compensator} as
\begin{equation}
\dot{\eta}_i(t)=S\eta_i(t) + \mu d_i\sum\nolimits_{j\in\mathcal{N}_i}{(\hat{\eta}_j(t)-\hat{\eta}_i(t))},\label{compensatorcancel}
\end{equation}
where $d_i$ is chosen such that all the eigenvalues of $DH$ are real, positive, and distinct with $D=\textnormal{diag}(d_1,\cdots,d_N)$. The existence of $D$ has been shown in \cite[Lemma 4]{wang2020adaptive} under Assumption \ref{ass0}. Then, system \eqref{solution1} can be correspondingly modified into the following system:
\begin{equation}\label{solution1cancel}
\tilde{\eta}(t_{k+1})=F\left(\mu\right)\tilde{\eta}(t_k),
\end{equation}
where
\begin{align*}
F\left(\mu\right)&=I_N\otimes e^{Sh}-\mu DH\otimes \int_{0}^{h}e^{S\tau}d\tau.
\end{align*}
By replacing $H$ with $DH$ in Theorem \ref{radio}, we have the following corollary.
\begin{Corollary} Under Assumption \ref{ass0}, there exists a $\mu>0$ such that $F\left(\mu\right)$ is Schur if and only if $h$ satisfies the conditions below:
\begin{enumerate}[(I)]
\item \label{noncodn1} Phase condition:
   $$(\phi_{q}(h) -\theta_q)\in(-\frac{\pi}{2}, \frac{\pi}{2}),~~~~q\in \mathcal{Q}_1\cup Q_3~~\textnormal{and}~~i\in\mathcal{V}_F;$$
\item \label{noncodn2}  Magnitude condition: $$e^{2\re(\lambda_q)h}<\csc^2(\phi_{q}(h) -\theta_q),~~~~q\in \mathcal{Q}_1~~\textnormal{and}~~i\in\mathcal{V}_F;$$
\item \label{noncodn2+} Non-zero spectral mapping: $$h\neq \frac{2\kappa\pi}{\im(\lambda_q)},~~~~q\in \mathcal{Q}_3~~\textnormal{and}~~i\in\mathcal{V}_F,$$ where $\kappa \in \mathds{Z}$ and $\frac{\kappa}{\im(\lambda_q)}>0;$
\item \label{noncodn3} $\bigcap\limits_{q\in \mathcal{Q},~i\in \mathcal{V}_F}(\underline{\mu}_{i,q},\bar{\mu}_{i,q})\neq \emptyset$,
\end{enumerate}
where
\begin{align}
\bar{\mu}_{i,q}=&{\textstyle |\lambda_q|\frac{\cos(\phi_{q}(h) -\theta_q)+\sqrt{e^{-2\re(\lambda_q)h}-\sin^2(\phi_{q}(h) -\theta_q)}}{\lambda_i(DH)e^{-\re(\lambda_q)h}\sqrt{V_{q}^2(h)+U_{q}^2(h)}}},\nonumber\\
\underline{\mu}_{i,q}=&{\textstyle |\lambda_q|\frac{\cos(\phi_{q}(h) -\theta_q)-\sqrt{e^{-2\re(\lambda_q)h}-\sin^2(\phi_{q}(h) -\theta_q)}}{\lambda_i(DH)e^{-\re(\lambda_q)h}\sqrt{V_{q}^2(h)+U_{q}^2(h)}}}.\label{solutionmu}
\end{align}
\end{Corollary}
The proof is similar with the proof of Theorem \ref{radio} by replacing $\lambda_i(H) $ and $\theta_i(H)$ with $\lambda_i(DH)$ and $\theta_i(DH)=0$, respectively. Thus, the proof is omitted.
}
\subsection{The convergence of the time- and event-triggered error dynamics}
Regarding the convergence of the estimation errors $\tilde{\eta}_i$ with the triggering function in (\ref{setrigger}), we have the following theorem.
\begin{Theorem}\label{lemma4} Under Assumption \ref{ass0}, there exists a $\mu>0$ such that the estimation error $\tilde{\eta}(t)$ of the mixed time- and event-triggered distributed observer converges for any event-triggered mechanism in the form of \eqref{setrigger} if $h$ satisfies the conditions in Theorem \ref{radio}. Furthermore, all feasible $\mu>0$ are in the set $\bigcap\limits_{q\in \mathcal{Q},~i\in \mathcal{V}_F}(\underline{\mu}_{i,q},\bar{\mu}_{i,q})$, where $\underline{\mu}_{i,q}$ and $\bar{\mu}_{i,q}$ are defined in \eqref{solutionmu}. {\myr Moreover, if $f_i(t_k)$ tends to zero (exponentially), so does $\tilde{\eta}(t)$.}
\end{Theorem}
\begin{proof}
If $h$ satisfies the conditions in Theorem \ref{radio}, then, from Theorem \ref{radio}, there exists a $\mu\in \bigcap\limits_{q\in \mathcal{Q},~i\in \mathcal{V}_F}(\underline{\mu}_{i,q},\bar{\mu}_{i,q})$ such that $F(\mu)$ is Schur. The triggering function in \eqref{setrigger} enforces
$\left\| \bar{\eta}(t_k)\right\|
\leq \sqrt{N} f(t_k)$ after applying the updating law in (\ref{exosystem1}),
where $f(t_k)=\max\left\{f_1(t_k),\cdots,f_N(t_k)\right\}$. Based on this and the convergence of $f_i(t_k)$, we have that
\begin{equation}\label{yitabarc0}
\lim\limits_{t_k\rightarrow\infty}\bar{\eta}(t_k)=0.
\end{equation}
Applying \citep*[Lemma~3.1]{yan2017cooperative} together with this to \eqref{solution1} yields that
\begin{equation}\label{yitatildec0}
\lim\limits_{t_k\rightarrow\infty}\tilde{\eta}(t_k)=0.
\end{equation}
Then, from equation \eqref{reeq1}, we have that, $\forall t\in [t_{k},t_{k+1})$,
\begin{align}
\|\tilde{\eta}(t)\|\leq& e^{\|S\|h}\big[(1+\mu  h\|H\|) \|\tilde{\eta}(t_k)\|+ \mu  h\|H\|\|\bar{\eta}(t_k)\|\big].\nonumber
\end{align}
Therefore, if follows from (\ref{yitabarc0}) and (\ref{yitatildec0}) that $\lim\limits_{t\rightarrow\infty}\tilde{\eta}(t)=0$.

On the contrary, if $\lim\limits_{t\rightarrow\infty}\tilde{\eta}(t)=0$, then $F(\mu)$ must be Schur, and thus $h$ satisfies the conditions in Theorem \ref{radio}.

{\myr In particular, since $f_i(t_k)$ tends to zero(exponentially), from Lemma.~\ref{lemmaaymmk} and equation \ref{solution1}, we have $\lim\limits_{t_k\rightarrow\infty}\tilde{\eta}(t_k)=0$ exponentially. Hence, $\tilde{\eta}(t)$ tends to zero(exponentially).}
\end{proof}
\section{Inter-event analysis}\label{intereventstep}
{\myr Just as \cite{liu2015event} pointed out, the numerical simulation's step-length guarantees strictly positive minimum inter-event times and excludes Zeno behavior. While \cite{liu2015event} still observed that the inter-event times converge to the step length of the numerical simulation in finite time. Hence, it is of practical interest and deserves more effort to investigate the analytical characterization of the relationship among sampling, event triggering, and inter-event behavior. }

Under Assumption \ref{ass0}, there exists a $\mu$ such that $F(\mu)$ is Schur for any $h$ satisfying the conditions in Theorem \ref{radio}. For such a matrix, there are constants $\beta(h,\mu)>0$ and $0\leq \gamma(h,\mu)<1 $ such that $\|F^{k}(\mu)\|\leq \beta(h,\mu) \gamma^k(h,\mu)$, where $\gamma(h,\mu)$ is the spectral radius of $F(\mu)$. In the remaining analysis, $\mu$ is supposed to be fixed, and $h$ and $\mu$ in $\beta(h,\mu)$ and $\gamma(h,\mu)$ are omitted for notational simplicity. To analyze the number of steps between two consecutive triggering instants, the triggering function for each follower $i$ is assumed to be bounded by exponential functions, that is, \begin{equation}\label{setriggerex}\sigma_m e^{-\alpha t_k }\leq  f_i(t_k)\leq \sigma_M e^{-\alpha t_k },\end{equation}where $\sigma_m >0$, $\sigma_M>0$, and $\alpha>0$, for $i\in \mathcal{V}_F$. 
Given the triggering instant $t_{l}^{i}$, suppose that the next triggering instant is $t_k=t_{l}^{i}+s^i_{l} h$. Then, the solution of the system in \eqref{compensator} during the time interval $[t_{l}^{i},t_{l}^{i}+s^i_{l} h)$ can be evaluated as
\begin{equation}
\eta_i(t_{k}) = {\textstyle  e^{Ss^i_{l}h}\eta_i(t_{l}^{i})+ \mu\sum\nolimits_{r=1}^{s^i_{l}}\int_{t_{k-r}}^{t_{k-r+1}}e^{S (t_k-\tau)}\hat{\eta}_{ei}(t_{k-r})d\tau}.\nonumber
\end{equation}
Consider equation \eqref{exosystem1} during the time interval $[t_{l}^{i},t_{l}^{i}+s^i_{l} h)$. Then, we have that
\begin{align}\label{solutionerrxx2}
\bar{\eta}_i(t_k^{-}) =&{\textstyle e^{S s^i_{l} h}\bar{\eta}_i(t_{l}^{i})}\nonumber\\
&{\textstyle+ \mu \sum\nolimits_{r=1}^{s^i_{l}}\int_{t_{k-r}}^{t_{k-r+1}}e^{S (t_k-\tau)}\hat{\eta}_{ei}(t_{k-r})d\tau}.\end{align}
where $\bar{\eta}_i(t_k^{-})$ denotes the value of $\bar{\eta}_i(t)$ at $t_k$ before applying the updating law in (\ref{exosystem1}). In the following theorem, we establish a relationship among the inter-event steps $s^i_{l}$, the sampling period $h$, and the decay rate of the bounding functions $\alpha$. Before proceeding, we define some variables below
\begin{align}
\chi_1\triangleq & \mu\|H\|\sqrt{N}\sigma_M,~~~~\chi_2(h)\triangleq{\textstyle  \frac{ \chi_1 \beta \|G(\mu)\|e^{\alpha h}}{1-\gamma e^{\alpha h}}},\nonumber\\
\chi_3(k,h)\triangleq&   {\textstyle \mu\|H\| \beta\|\tilde{\eta}(t_{0})\| (\gamma e^{\alpha h})^{k} +\chi_2(h)[1-(\gamma e^{\alpha h})^{k}]}.\nonumber\end{align}
\subsection{Inter-event steps}
\begin{Theorem}\label{cor5} Under Assumption \ref{ass0}, $s^i_{l}$ has the following properties for any event-triggered mechanism in the form of \eqref{setrigger} with the triggering function in (\ref{setriggerex}):
\begin{enumerate}
  \item (Time-varying bound) there exists a positive $s(k,h)$ such that
\begin{equation}\label{sitriggersin}
s^i_{l}> s(k,h),~~~~k\in \mathds{N};
\end{equation}
  \item (Asymptotic bound) there exists a positive constant $s^{*}\left(h\right)$ such that
 \begin{align}\label{sitriggersin2}
\lim_{k\rightarrow\infty}s(k,h)= \left\{\begin{array}{cc}
  s^{*}\left(h\right),~~~~&e^{\alpha h}\gamma<1,\\
0,~~~~&e^{\alpha h}\gamma\geq 1,
 \end{array}\right.
 \end{align}
\end{enumerate}
where $s(k,h)$ and $s^*\left(h\right)$ are given by the following equations in $s(k,h)$ and $s^{*}$, respectively,
 \begin{subequations}\begin{align}
\frac{e^{\|S\|s(k,h) h}-1}{\|S\|}  =& \frac{\sigma_m }{\gamma^{-s(k,h)}\chi_3(k,h)+ \chi_1 e^{\alpha hs(k,h)} },\label{solutionS}\\
\frac{e^{\|S\|s^{*}(h) h}-1}{\|S\|}  =& \frac{\sigma_m}{\gamma^{-s^{*}(h)}\chi_2(h) + \chi_1 e^{\alpha hs^{*}(h)} }.\label{chiii}
\end{align} \end{subequations}
\end{Theorem}
\begin{proof}
Simple calculation from \eqref{solution1} gives that
\begin{equation}
\|\tilde{\eta}(t_k)\|\leq{\textstyle \beta\gamma^k\|\tilde{\eta}(t_0)\|+\sum\limits_{r=0}^{k-1}\beta\gamma^{k-r-1}\|G(\mu)\|\|\bar{\eta}(t_r)\|}.\nonumber
\end{equation}
The event-triggered mechanism in \eqref{setrigger} enforces $\|\bar{\eta}(t_r)\|\leq  \sqrt{N}\sigma_M e^{-\alpha h r}$ after applying the updating law in (\ref{exosystem1}).
Then,
\begin{equation}
\|\tilde{\eta}(t_{k})\|
\leq{\textstyle \beta\gamma^k\|\tilde{\eta}(t_{0})\|+\beta\sigma_M\sqrt{N}\|G(\mu)\|\gamma^{k-1}\sum\limits_{r=0}^{k-1} \frac{e^{-\alpha hr}}{\gamma^r}}\nonumber
\end{equation}
From equations \eqref{relaH}, it follows that
\begin{align}\label{chiexponen}
\left\|\hat{\eta}_{ei}\left(t_{k}\right)\right\|\leq \left\|\hat{\eta}_{e}\left(t_{k}\right)\right\|
\leq \|H\|\left(\|\tilde{\eta}(t_k)\|+\|\bar{\eta}(t_k)\|\right).
\end{align}
Obviously, $\bar{\eta}_i(t_{l}^{i})=0$, which implies together with \eqref{solutionerrxx2} that
\begin{equation}\label{sorxx3}
\bar{\eta}_i(t_k^{-}) ={\textstyle \mu \sum\nolimits_{r=1}^{s^i_{l}}\int_{t_{k-r}}^{t_{k-r+1}}e^{S (t_k-\tau)}\hat{\eta}_{ei}(t_{k-r})d\tau}.
\end{equation}
For any positive $r\leq s^i_{l}$, it is noted that $\gamma^{k-r-1}\leq \gamma^{k-s^i_{l}-1}$, $\gamma^{k-r}\leq \gamma^{k-s^i_{l}}$, and $\sum\nolimits_{v=0}^{k-r-1}\frac{e^{-\alpha hv}}{\gamma^v}\leq \sum\nolimits_{v=0}^{k-1}\frac{e^{-\alpha hv}}{\gamma^v}$.
Then, from equations \eqref{chiexponen}, \eqref{sorxx3}, and $\|\bar{\eta}(t_r)\|\leq  \sqrt{N}\sigma_M e^{-\alpha h r}$, it follows that
\begin{align}\label{solutionerrxx3}
\|\bar{\eta}_i(t_{k}^{-})&\|\leq{\textstyle \mu \sum\nolimits_{r=1}^{s^i_{l}}\int_{t_{k-r}}^{t_{k-r+1}}\|e^{S (t_k-\tau)}\|\left\|\hat{\eta}_{e}\left(t_{k-r}\right)\right\|d\tau}\nonumber\\
\leq& {\textstyle \int_{0}^{s^i_{l}h} \|e^{S\tau}\| d\tau\big[\gamma^{-s^i_{l}}\chi_3(k,h) + \chi_1 e^{\alpha h s^i_{l}} \big]e^{-\alpha h k}}\\
\leq& {\textstyle \int_{0}^{s^i_{l}h} e^{\|S\|\tau} d\tau\big[\gamma^{-s^i_{l}}\chi_3(k,h) + \chi_1 e^{\alpha h s^i_{l}} \big]e^{-\alpha h k}}.\nonumber
\end{align}
On the other hand,
\begin{equation}\label{leftinexxsin}\|\bar{\eta}_i(t_k^{-})\| =\|\bar{\eta}_i(t_{l}^{i}+s^i_{l} h)\|> f_{i}(t_k) \geq\sigma_m e^{-\alpha h k }.\end{equation}
Combining \eqref{solutionerrxx3} and \eqref{leftinexxsin}  yields that
\begin{equation}\label{solutionofS} \frac{e^{\|S\|s^i_{l}h}-1}{\|S\|} > \frac{\sigma_m }{ \gamma^{-s^i_{l}}\chi_3(k,h)+ \chi_1 e^{\alpha hs^i_{l}} },~~~~\|S\|\neq 0,
\end{equation}
When $\|S\|=0$, it is easy to show the left side of (\ref{solutionofS}) can be evaluated using the liming process $\lim\limits_{\|S\|\rightarrow 0}\frac{e^{\|S\|s^i_{l}h}-1}{\|S\|}=s^i_{l}h$. Obviously, the left hand side of \eqref{solutionofS} is increasing in $s^i_{l}$, and the right hand side of \eqref{solutionofS} is decreasing in $s^i_{l}$, since $e^{\|S\|s^i_{l}h} $, $\gamma^{-s^i_{l}}$, and $ e^{\alpha hs^i_{l}}$ are increasing in $s^i_{l}$. Also, when $s^i_{l}=0$, the left hand side equals $0$ and the right hand side is $\frac{\sigma_m }{\chi_1 + \chi_3(k,h) }> 0$. As a result, there exists a unique positive $s(k,h)$ such that $s^i_{l}>s(k,h)$, where $s(k,h)$ can be found by solving \eqref{solutionS}.

The solution to \eqref{solutionS} is dependent on $h$ and $k$, making the lower bound time-varying. To derive a time-invariant bound independent of $k\in \mathds{N}$, the solution $s(k,h)>0$ to \eqref{solutionS} is discussed below based on the $e^{\alpha h}\gamma$.\\
\textbf{Case 1. $e^{\alpha h}\gamma<1$:} From \eqref{solutionS}, $\lim_{k\rightarrow \infty}\chi_3(k,h)=\chi_2(h)$.
Then, there exists a unique time-invariant bound $s^*$ such that $\lim\limits_{k\rightarrow \infty}s(k,h)= s^*\left(h\right)$, where $s^*\left(h\right)$ is given by \eqref{chiii}.\\
\textbf{Case 2. $e^{\alpha h}\gamma \geq 1$:} From \eqref{solutionS}, $\lim_{k\rightarrow \infty}\chi_3(k,h)= \infty$ and, therefore, $\lim\limits_{k\rightarrow \infty}s(k,h)= 0$ and ${\textstyle \frac{e^{\|S\|s(k,h) h}-1}{\|S\|}}\geq 0$.
\end{proof}
If $S=0$ or $S^T=-S$ implying $\|e^{St}\|=1$, we further have the following corollary.
\begin{Corollary}\label{corsin7} Under Assumption \ref{ass0}, $s^i_{l}$ has the following properties for any event-triggered mechanism in the form of \eqref{setrigger} with the triggering function in (\ref{setriggerex}):
\begin{enumerate}
  \item there exists a positive $s(k,h)$ such that
\begin{equation}\label{corsitriggersin}
s^i_{l}> s(k,h),~~~~k\in \mathds{N};
\end{equation}
  \item there exists a positive constant $s^{*}$ such that
 \begin{align}\label{corsitriggersin2}
\lim_{k\rightarrow\infty}s(k,h)=  \left\{\begin{array}{cc}
  s^{*}\left(h\right),~~~~&e^{\alpha h}\gamma<1,\\
0,~~~~&e^{\alpha h}\gamma\geq 1,
 \end{array}\right.
 \end{align}
\end{enumerate}
where $s(k,h)$ and $s^*\left(h\right)$ are given by the following equations in $s(k,h)$ and $s^*$, respectively,
 \begin{align}
 s(k,h) h=& \frac{\sigma_m }{\gamma^{-s(k,h)}\chi_3(k,h)+ \chi_1 e^{\alpha hs(k,h)} },\nonumber\\
s^*(h) h=& \frac{\sigma_m}{\gamma^{-s^{*}(h)}\chi_2(h) +\chi_1  e^{\alpha hs^{*}(h)}}.\nonumber
\end{align}
\end{Corollary}
\begin{Remark}Intuitively, $s^i_{l}$ would be smaller if the bounding function $f_i(t_k)$ decays faster than the error dynamics, that is, $\gamma \geq e^{-\alpha h}$, leading to more frequent event triggering. The results in Theorem \ref{cor5} indicate that this may be true as $s(k,h)$ is a decreasing function in $k$ and lower bounded by $0$, which together with $s^i_{l}\in \mathds{Z}^{+}$ implies that $s^i_{l}\geq1$. By contrast, if the bounding function $f_i(t_k)$ decays slower than the error dynamics, that is, $\gamma< e^{-\alpha h}$, there would be less frequent triggering as shown in (\ref{sitriggersin2}) and (\ref{corsitriggersin2}).
{\myr{Theorem \ref{corsin7} can be used to explain the phenomenon observed in \cite{liu2015event} that the inter-sampling times can converge to the step size of simulation in a short time.}}
\end{Remark}
\subsection{Inter-event time}
Define $\tau_l^{i}\triangleq s_{l}^{i}h$ as the inter-event time. Then, by following the same line as used in the proof of Theorem \ref{cor5}, we can establish the following result.
\begin{Theorem}\label{corsin8} Under Assumption \ref{ass0}, $\tau_l^{i}$ has the following properties for any event-triggered mechanism in the form of \eqref{setrigger} with the triggering function in (\ref{setriggerex}):
\begin{enumerate}
  \item (Time-varying bound) there exists a positive $\tau_d(k,h)$ such that
\begin{equation}\label{timessitriggersin}
\tau_l^{i}> \tau_d(k,h),~~~~k\in \mathds{N};
\end{equation}
  \item (Asymptotic bound), there exists a positive constant $\tau_d^*$ such that
 \begin{align}\label{timessitriggersin2}
\lim_{k\rightarrow\infty}\tau_d(k,h)= \left\{\begin{array}{cc}
  \tau_d^*\left(h\right),~~~~&e^{\alpha h}\gamma<1,\\
0,~~~~&e^{\alpha h}\gamma\geq 1,
 \end{array}\right.
 \end{align}
\end{enumerate}
where $\tau_d(k,h)$ and $\tau_d^*\left(h\right)$ are given by the following equations in $\tau_d(k,h)$ and $\tau_d^*(h)$, respectively,
\begin{subequations} \begin{align*}
 \frac{e^{\|S\|\tau_d(k,h)}-1}{\|S\|}  =&\frac{\sigma_m }{\gamma^{-\frac{\tau_d(k,h)}{h}}\chi_3(k,h)+ \chi_1 e^{\alpha \tau_d(k,h)} },\\
\frac{e^{\|S\|\tau_d^*(h)}-1}{\|S\|}  =& \frac{\sigma_m}{\gamma^{-\frac{\tau_d^*(h)}{h}}\chi_2(h) +\chi_1 e^{\alpha \tau_d^*(h)} }
\end{align*}\end{subequations}
\end{Theorem}

\begin{Remark}From the proof of Theorem \ref{radio}, it is noted that $\gamma^2=\alpha_{i,q} \mu^2+\beta_{i,q}\mu+\gamma_{i,q}+1$ for some $i\in \mathcal{V}_F$ and $q\in \mathcal{Q}$, where $\alpha_{i,q}$, $\beta_{i,q}$, and $\gamma_{i,q}$ are defined in (\ref{alphaq3q2q4}). Based on this, it can be verified that $\lim_{h\rightarrow0}\left(\gamma^2-1\right)=0$ and $\lim_{h\rightarrow0}\frac{\gamma^2-1}{h}=2\re(\lambda_q)-2\re(\lambda_i)\mu$ according to \eqref{abc}. Then, we have that
\begin{align*}
\lim_{h\rightarrow0}\left(\gamma^2\right)^{-\frac{1}{2h}}=&\lim_{h\rightarrow0}\left[\left(\gamma^2-1+1\right)^{\frac{1}{\gamma^2-1}}\right]^{-\frac{\gamma^2-1}{2h}}\nonumber\\
=&e^{\re(\lambda_i)\mu-\re(\lambda_q)}.
\end{align*}
Also, $\lim_{h\rightarrow0}\chi_2(h)=\chi_2^{*}$ for some positive $\chi_2^{*}$. As a consequence, $\lim_{h\rightarrow0}\tau_d^{*}(h)=\tau^*$, where $\tau^*$ is the solution to
$$\frac{e^{\|S\|\tau^*}-1}{\|S\|}  = \frac{\sigma_m}{\chi_2^{*}e^{(\re(\lambda_i)\mu-\re(\lambda_q))\tau^*} +\chi_1 e^{\alpha \tau^*} }.$$
Therefore, from this and \eqref{timessitriggersin}, it follows that, when $h$ is sufficiently small, $s_{l}^{i}>\tau^*/h$. This shows that the inter-event steps and the sampling period are in a relation of almost negative proportionality, indicating that an extremely small sampling period may result in less frequent event triggering. Further analysis on the relationship between $h$ and $s^*\left(h\right)$ (or $\tau_d^*\left(h\right)$) and related optimization will be reported elsewhere.
\end{Remark}

\section{Numerical Examples}\label{section6}
{\myr{In the following Examples 1-3, the follower dynamics will be omitted to reveal better the interactions among sampling periods, topologies, reference signals, and related observability. We will establish the analytical characterization of
the relationship among sampling, event triggering, and inter-
event behavior.
}}
\subsection{Example 1: Time-triggered observers with sinusoidal signal}
\begin{figure}[htp]
\begin{center}
\begin{tikzpicture}[transform shape]
    \centering%
    \node (0) [circle, draw=red!20, fill=red!60, very thick, minimum size=7mm] {\textbf{0}};
    \node (1) [circle, right=of 0, draw=blue!20, fill=blue!60, very thick, minimum size=7mm] {\textbf{1}};
    \node (2) [circle, right=of 1, draw=blue!20, fill=blue!60, very thick, minimum size=7mm] {\textbf{2}};
    \node (3) [circle, right=of 2, draw=blue!20, fill=blue!60, very thick, minimum size=7mm] {\textbf{3}};
    \draw[ very  thick,->,  left] (0) edge (1);
    \draw[ very  thick,->,  bend left] (3) edge (2);
    \draw[ very  thick,->, bend left] (1) edge (3);
       \draw[ very  thick,->, bend left] (2) edge (1);
          \draw[ very  thick,->, bend left] (2) edge (3);
\end{tikzpicture}
\end{center}
\caption{Communication topology ${\mathcal{G}}$}
\label{fig1numex}
\end{figure}
 \begin{figure}[ht]
  \centering\setlength{\unitlength}{0.75mm}
  \epsfig{figure=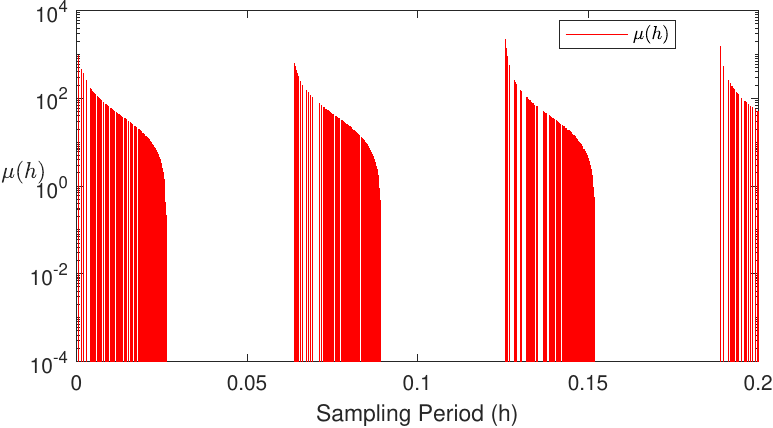,width=\linewidth}
  \caption{The non-pathological sampling period and $\mu$}\label{hmuvalue}
\end{figure}

 \begin{figure}[ht]
  \centering\setlength{\unitlength}{0.75mm}
  \epsfig{figure=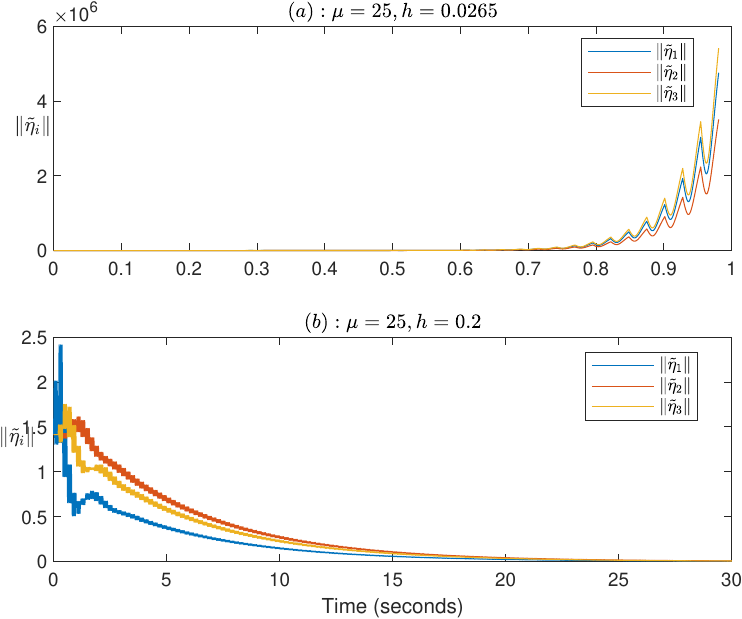,width=\linewidth}
  \caption{Estimation errors with different parameters}\label{figsa}
\end{figure}

Consider distributed system shown in Figure~\ref{fig1numex} with \begin{equation}S=\left[
     \begin{matrix}
       0 & 100 \\
       -100 & 0 \\
     \end{matrix}
   \right]~~~\textnormal{and}~~~H=\left[
                                                                                      \begin{matrix}
                                                                                       2   &  -1 & 0 \\
                                                                                        0  &  1 & -1  \\
                                                                                         -1 &  -1 &  2 \\
                                                                                      \end{matrix}
                                                                            \right].\nonumber\end{equation}

The eigenvalues of $H$ are $\{2.4196 \pm 0.6063j, 0.1067\}$, implying $\theta_i\in \{0,\pm 0.2455\}$. It can be easily calculated that all the non-pathological sampling periods induced $\mu$ shown in Figure~\ref{hmuvalue}.

It can be easily verified that $h=0.0265$ is a pathological sampling period over the graph from \eqref{h4Q3}, and thus $\mu$ does not exist for this sampling period. The simulation in Figure~\ref{figsa}.(a) shows that the estimation errors don't converge to 0 as we analyzed.

$h=0.2$ is a non-pathological sampling period according to Corollary \ref{corral}, and the corresponding feasible interval of $\mu$ is $\bigcap\limits^{}_{q\in \{1,2\}, i\in\{1,2,3\}}(\underline{\mu}_{i,q},\bar{\mu}_{i,q})=(0,50.2333)$. We choose $\mu=25$, and the estimation errors converge to 0 as depicted in Figure~\ref{figsa}. (b). This simple example shows that even a small sampling period may not guarantee convergence due to the existence of pathological sampling periods.

\subsection{Example 2: Time-triggered observers with unbounded signal}
\begin{figure}[htp]
\begin{center}
\begin{tikzpicture}[transform shape]
    \centering%
    \node (0) [circle, draw=red!20, fill=red!60, very thick, minimum size=7mm] {\textbf{0}};
    \node (1) [circle, right=of 0, draw=blue!20, fill=blue!60, very thick, minimum size=7mm] {\textbf{1}};
    \node (2) [circle, right=of 1, draw=blue!20, fill=blue!60, very thick, minimum size=7mm] {\textbf{2}};
    \node (3) [circle, right=of 2, draw=blue!20, fill=blue!60, very thick, minimum size=7mm] {\textbf{3}};
    \draw[ very  thick,->,  left] (0) edge (1);
    \draw[ very  thick,->,  bend left] (3) edge (2);
    \draw[ very  thick,->, bend left] (1) edge (2);
       \draw[ very  thick,->, bend left] (2) edge (1);
          \draw[ very  thick,->, bend left] (2) edge (3);
\end{tikzpicture}
\end{center}
\caption{Communication topology ${\mathcal{G}}$}
\label{fig1numex2spec}
\end{figure}
 \begin{figure}[ht]
  \centering\setlength{\unitlength}{0.75mm}
  \epsfig{figure=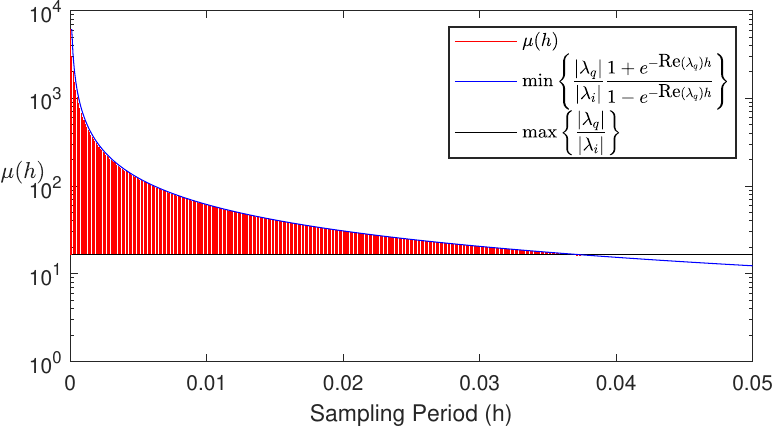,width=\linewidth}
  \caption{The non-pathological sampling period and $\mu$}\label{hmuvalue2spec}
\end{figure}
 \begin{figure}[ht]
  \centering\setlength{\unitlength}{0.75mm}
  \epsfig{figure=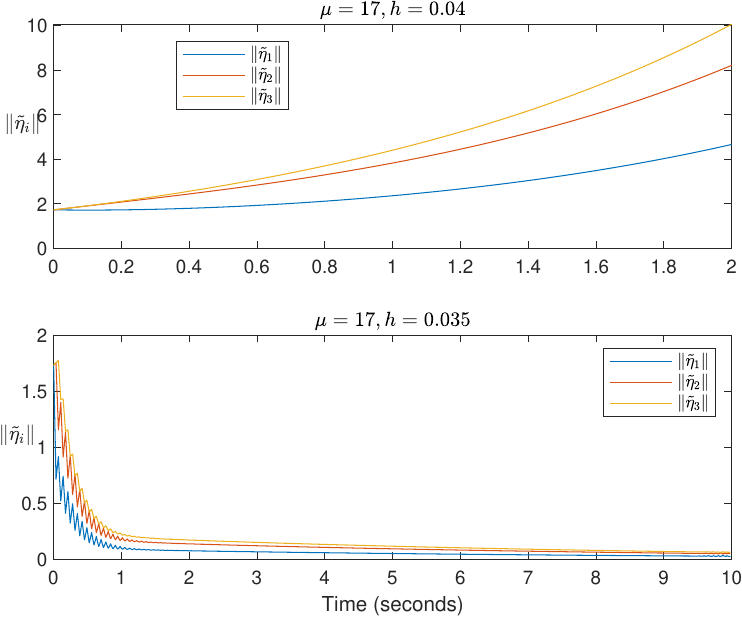,width=\linewidth}
  \caption{Estimation errors with different parameters}\label{ex2figsaspec}
\end{figure}

{\myr We now consider distributed system shown in Figure~\ref{fig1numex2spec} with \begin{equation}S=\left[
                                                                                      \begin{matrix}
                                                                                       2   &  -1 & 0 \\
                                                                                       -1  &  2 & -1  \\
                                                                                         0 &  -1 &  1 \\
                                                                                      \end{matrix}
                                                                            \right]~~~\textnormal{and}~~~H=\left[
                                                                                      \begin{matrix}
                                                                                       2   &  -1 & 0 \\
                                                                                       -1  &  2 & -1  \\
                                                                                         0 &  -1 &  1 \\
                                                                                      \end{matrix}
                                                                            \right].\nonumber\end{equation}
The eigenvalues of $H$ and $S$ are $\{0.1981,  1.555, 3.247\}$, implying $\theta_i\equiv\theta_q\equiv0$. It can be easily calculated that all the non-pathological sampling periods induced $\mu$ shown in Figure~\ref{hmuvalue2spec} according to Remark \ref{Speremark3}.

It can be easily verified that $h=0.04$ is a pathological sampling period over the graph, and $\bigcap\limits^{}_{q\in \{1,2\}, i\in\{1,2,3\}}(\underline{\mu}_{i,q},\bar{\mu}_{i,q})=\emptyset$. Hence $\mu$ does not exist for this sampling period. The simulation in Figure~\ref{ex2figsaspec}.(a) shows that the estimation errors don't converge to 0 as we analyzed.

$h=0.035$ is a non-pathological sampling period according to Theorem \ref{radio}, and the corresponding feasible interval of $\mu$ is $\bigcap\limits^{}_{q\in \{1,2\}, i\in\{1,2,3\}}(\underline{\mu}_{i,q},\bar{\mu}_{i,q})=(16.3937, 17.5988)$. We choose $\mu=17$, and the estimation errors converge to 0 as depicted in Figure~\ref{ex2figsaspec}. (b). This simple example shows that even a small sampling period may not guarantee convergence due to the existence of pathological sampling periods.}

\subsection{Example 3: Mixed time- and event-triggered observers}
\begin{figure}[htbp]
\begin{center}
\begin{tikzpicture}[transform shape]
    \centering%
    \node (0) [circle,draw=red!20, fill=red!60, very thick, minimum size=7mm]{\textbf{0}};
    \node (1) [circle, right=of 0, draw=blue!20, fill=blue!60, very thick, minimum size=7mm] {\textbf{1}};
    \node (2) [circle, right =of 1, draw=blue!20, fill=blue!60, very thick, minimum size=7mm] {\textbf{2}};
    \node (3) [circle, right=of 2, draw=blue!20, fill=blue!60, very thick, minimum size=7mm] {\textbf{3}};
    \node (4) [circle, right=of 3, draw=blue!20, fill=blue!60, very thick, minimum size=7mm] {\textbf{4}};
    \draw [ very  thick,->, bend right] (0) edge (2);
    \draw [ very  thick,->, bend left] (0) edge (1);
     \draw [ very  thick,->, bend right] (4) edge (1);
     \draw [ very  thick,->, bend right] (1) edge (2);
      \draw [ very  thick,->, bend right] (4) edge (3);
      \draw [ very  thick,->, bend right] (3) edge (4);
      \draw [ very  thick,->, bend right] (3) edge (2);
      \draw [ very  thick,->, bend right] (2) edge (3);
        \draw [ very  thick,->, bend right] (2) edge (4);
    \end{tikzpicture}
\end{center}
\caption{ Communication topology ${\mathcal{G}}$}
\label{fig1}
\end{figure}
In this example, we consider a linear distributed system composed of one leader and four followers in the form of \eqref{leader} with
$S=\left[
           \begin{smallmatrix}
               0 & 1 \\
              -1 & 0 \\
             \end{smallmatrix}
           \right]$. It can be easily calculated that $\|e^{St}\|=1$.
 $v(t)\in \mathds{R}^{2}$ can be measured by followers $1$ and $2$ as shown in Figure~\ref{fig1}. The matrix $H$ for this example can be determined as
\begin{align}H=\left[
       \begin{matrix}
         2 & 0 & 0 & -1 \\
         -1 & 3 & -1 & 0 \\
         0 & -1 & 2 & -1 \\
         0 &-1 & -1& 2 \\
       \end{matrix}
     \right].\nonumber
\end{align}

\begin{figure}[ht]
  \centering\setlength{\unitlength}{0.75mm}
 \includegraphics[width=\linewidth]{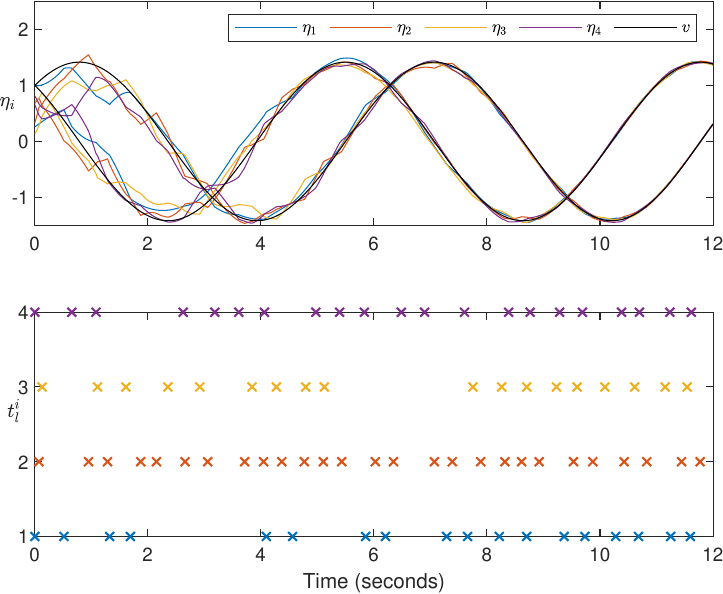}
  \caption{Estimation performance of all followers.}\label{figcase3ese}
\end{figure}

\begin{figure}[ht]
  \centering\setlength{\unitlength}{0.75mm}
  \epsfig{figure=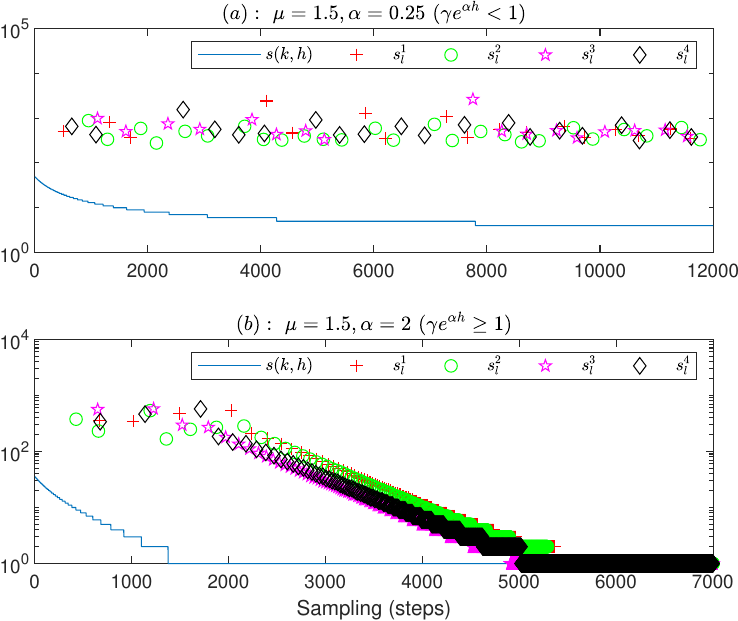,width=3.4in}
  \caption{Inter-event steps.}\label{figcase1tre}
\end{figure}
 We design a mixed time- and event-triggered distributed observers composed of \eqref{compensator} with the event-triggered condition in \eqref{setrigger} where $f_i(t_k)=\sigma_i e^{-\alpha_i t_k}$. Simulations are conducted with the following parameters: $\mu=1.5$, $\sigma_i=1$ and $h=0.001$. It can be verified that $\gamma=0.9994$ and $\beta=1.0004$. Figure~\ref{figcase3ese} shows the estimation performance of all followers with $\alpha =0.25$. The inter-event steps $s_{l}^{i}$ and $s(k,h)$ for the case $\gamma e^{\alpha h}<1$ with $\alpha =0.25$ and for the case $\gamma e^{\alpha h}\geq1$ with $\alpha =2$ are shown in Figure~\ref{figcase1tre}.(a) and Figure~\ref{figcase1tre}.(b), respectively. It can be seen from the figure that slow triggering functions ($\gamma e^{\alpha h}<1$) will probably increase the inter-event steps $s^i_{l}$, and fast triggering function ($\gamma e^{\alpha h}\geq1$) may lead to the degradation of the event-triggered system to a pure sampled-data system.

\section{Conclusions}\label{conlu}
In this paper, the distributed state estimation problem for linear distributed systems has been investigated using mixed time- and event-triggered observers.
A distributed observer with time-triggered observations has been proposed to reconstruct the leader's state, and an auxiliary observer with event-triggered communication has been designed to reduce the information exchange among followers. Event triggering is based on locally sampled information and local computation, and Zeno behavior is naturally excluded.
All feasible sampling periods for the observability of the connected systems have been given, revealing some fundamental relationships among sampling periods, topologies, and system dynamics in distributed systems.
Time-varying and asymptotic bounds for the inter-event steps and time have been given, and their relationship with periodic sampling and event triggering has been shown.
%
Three numerical examples have been provided to show the effectiveness of the proposed control approach.

\section*{Appendix A}\label{appendixa}
\textbf{(1)}  Construct an auxiliary function $$f(x,y)=e^{2xh}-2e^{xh}\cos(yh)+1,$$
where $x,y\in\mathds{R}$. Taking Taylor series expansion around the origin yields that
\begin{align}
f(x,y)=&{\textstyle f(0,0)+x\frac{\partial f}{\partial x}(0,0)+y\frac{\partial f}{\partial y}(0,0)+\frac{1}{2}\big[x^2\frac{\partial^2 f}{\partial^2 x}(0,0)}\nonumber\\
&{\textstyle +2xy\frac{\partial^2 f}{\partial x \partial y}(0,0)+y^2\frac{\partial^2 f}{\partial^2 y}(0,0)\big]+o(\rho^2)}\nonumber\\
=& \rho^2h^2+o(\rho^2),\nonumber
\end{align}
where $\rho=\sqrt{x^2+y^2}$, and $o(\rho^2)$ is the Peano remainder such that $\lim_{\rho\rightarrow0} \frac{o(\rho^2)}{\rho^2}=0$. Noting that $V_{q}^2(h)+U_{q}^2(h)=f(\operatorname{Re}(\lambda_q),\operatorname{Im}(\lambda_q))$, we have that
\begin{align}
\frac{V_{q}^2(h)+U_{q}^2(h)}{|\lambda_q|^2}=&\frac{h^2|\lambda_q|^2+o(|\lambda_q|^2)}{|\lambda_q|^2}\nonumber\\
=&h^2 +\frac{o(|\lambda_q|^2)}{|\lambda_q|^2}.\nonumber
\end{align}
Hence, \begin{align}\label{uvh2}\lim_{\lambda_q\rightarrow0} \frac{V_{q}^2(h)+U_{q}^2(h)}{|\lambda_q|^2}=h^2. \tag{A.1}\end{align}

\textbf{(2)}
Construct two auxiliary functions
\begin{align}
g(x,y)=&\sin(yh)x- e^{xh}y+\cos(yh)y,\nonumber\\
p(x,y)=&\sin(yh)y+ e^{xh}x-\cos(yh)x,\nonumber
\end{align}
where $x,y\in\mathds{R}$. Taking Taylor series expansion around the origin yields that
\begin{align}
g(x,y)=&{\textstyle g(0,0)+x\frac{\partial g}{\partial x}(0,0)+y\frac{\partial g}{\partial y}(0,0)+\frac{1}{2}\big[x^2\frac{\partial^2 g}{\partial^2 x}(0,0)}\nonumber\\
&{\textstyle+2xy\frac{\partial^2 g}{\partial x \partial y}(0,0)+y^2\frac{\partial^2 g}{\partial^2 x}(0,0)\big]+o(\rho^2)}\nonumber\\
=& o(\rho^2),\nonumber\\
p(x,y)=&{\textstyle p(0,0)+x\frac{\partial p}{\partial x}(0,0)+y\frac{\partial p}{\partial y}(0,0)+\frac{1}{2}\big[x^2\frac{\partial^2 p}{\partial^2 x}(0,0)}\nonumber\\
&{\textstyle +2xy\frac{\partial^2 p}{\partial x\partial y}(0,0)+y^2\frac{\partial^2 p}{\partial^2 y}(0,0)\big]+o(\rho^2)}\nonumber\\
=& h\rho^2+o(\rho^2),\nonumber
\end{align}
where $\rho=\sqrt{x^2+y^2}$, and $o(\rho^2)$ is the Peano remainder such that $\lim_{\rho\rightarrow0} \frac{o(\rho^2)}{\rho^2}=0$. Since \begin{align}V_{q}(h)\re(\lambda_q)-U_{q}(h)\im(\lambda_q)=&g\left(\re(\lambda_q),\im(\lambda_q)\right),\nonumber\\ U_{q}(h)\re(\lambda_q)+V_{q}(h)\im(\lambda_q)=&p\left(\re(\lambda_q),\im(\lambda_q)\right),\nonumber
\end{align} we have that
\begin{align}
\frac{V_{q}(h)\re(\lambda_q)-U_{q}(h)\im(\lambda_q)}{|\lambda_q|^2}=&\frac{o(|\lambda_q|^2)}{|\lambda_q|^2},\nonumber\\
\frac{U_{q}(h)\re(\lambda_q)+V_{q}(h)\im(\lambda_q)}{|\lambda_q|^2}=&h+\frac{o(|\lambda_q|^2)}{|\lambda_q|^2}.\nonumber
\end{align}
Hence,
\begin{align}\label{vuoh}
\lim\limits_{\lambda_q\rightarrow0}\frac{V_{q}(h)\re(\lambda_q)-U_{q}(h)\im(\lambda_q)}{|\lambda_q|^2}=0,\nonumber\\
\lim\limits_{\lambda_q\rightarrow0}\frac{U_{q}(h)\re(\lambda_q)+V_{q}(h)\im(\lambda_q)}{|\lambda_q|^2}=h.\tag{A.2}
\end{align}
%
Next, we will prove that
\begin{description}
  \item[\textbf{(a)}] $(\phi_{q}(h)-\theta_q) \in (-\pi, \pi)$ when $\lambda_q\rightarrow0$;
  \item[\textbf{(b)}] $\lim\limits_{\lambda_q\rightarrow0}\sin(\phi_{q}(h)-\theta_q)=0$;
  \item[\textbf{(c)}] $\lim\limits_{\lambda_q\rightarrow0}\cos(\phi_{q}(h)-\theta_q)=1$.
\end{description}

\textbf{(a)} We discuss the following three cases based on $\lambda_q$:

\textbf{Case i. $\re(\lambda_q)=0$ and $\im(\lambda_q)\neq0$:} On one hand, as $\re(\lambda_q)=0$, $\lim\limits_{\im(\lambda_q)\rightarrow0^{+}}\theta_q=\frac{\pi}{2}$ or $\lim\limits_{\im(\lambda_q)\rightarrow0^{-}}\theta_q=-\frac{\pi}{2}$. On the other hand, when $\re(\lambda_q)=0$, $\phi_{q}(h)=\operatorname{Arg}\big((1-\cos(\im(\lambda_q)h)+\sin(\im(\lambda_q)h))j\big)$, which implies that $$\lim\limits_{\im(\lambda_q)\rightarrow0^{+}}\phi_{q}(h)=\frac{\pi}{2}~~~~\textnormal{or}~~~~\lim\limits_{\im(\lambda_q)\rightarrow0^{-}}\phi_{q}(h)=-\frac{\pi}{2}.$$

\textbf{Case ii. $\re(\lambda_q)\neq0$ and $\im(\lambda_q)=0$:} Since $\im(\lambda_q)=0$ and $\re(\lambda_q)\neq0$,
$\theta_{q}=0$ and $\phi_{q}(h)=\operatorname{Arg}\big((e^{\re(\lambda_q)h}-1)+0j\big)=0$.

\textbf{Case iii. $\re(\lambda_q)\neq$ and $\im(\lambda_q)\neq0$:} when $\im(\lambda_q)\rightarrow0$, $\sin(\im(\lambda_q)h)$ and $\im(\lambda_q)$ have the same sign. Hence, $\theta_{q} =\operatorname{Arg}(\re(\lambda_q)+\im(\lambda_q)j)$ and $\phi_{q}(h)=\operatorname{Arg}\big((e^{\re(\lambda_q)h}-\cos(\im(\lambda_q)h))+\sin(\im(\lambda_q)h)j\big)$ are either in the first/second quadrant or in the third/fourth quadrant at the same time.

Combining \textbf{Cases} i to iii gives that $(\phi_{q}(h)-\theta_q) \in (-\pi, \pi)$ when $\lambda_q\rightarrow0$.

\textbf{(b)} Applying the angle difference identity of the sine function yields that
\begin{align}
\sin(\phi_{q}&(h)-\theta_q)=\sin(\phi_{q}(h))\cos(\theta_q)-\cos(\phi_{q}(h))\sin(\theta_q)\nonumber\\
&=\frac{V_{q}(h)\re(\lambda_q)}{|\lambda_q|\sqrt{V_{q}^2(h)+U_{q}^2(h)}}-\frac{U_{q}(h)\im(\lambda_q)}{|\lambda_q|\sqrt{V_{q}^2(h)+U_{q}^2(h)}}\nonumber\\
&=\frac{V_{q}(h)\re(\lambda_q)-U_{q}(h)\im(\lambda_q)}{|\lambda_q|^2}\frac{|\lambda_q|}{\sqrt{V_{q}^2(h)+U_{q}^2(h)}}.\nonumber
\end{align}
From \eqref{uvh2} and \eqref{vuoh}, it follows that $\lim\limits_{\lambda_q\rightarrow0}\sin(\phi_{q}(h)-\theta_q)=0$. 

\textbf{(c)} Applying the angle difference identity of the cosine function yields that
\begin{align}
\cos(\phi_{q}&(h)-\theta_q)=\cos(\phi_{q}(h))\cos(\theta_q)+\sin(\phi_{q}(h))\sin(\theta_q)\nonumber\\
&=\frac{U_{q}(h)\re(\lambda_q)}{|\lambda_q|\sqrt{V_{q}^2(h)+U_{q}^2(h)}}+\frac{V_{q}(h)\im(\lambda_q)}{|\lambda_q|\sqrt{V_{q}^2(h)+U_{q}^2(h)}}\nonumber\\
&=\frac{U_{q}(h)\re(\lambda_q)+V_{q}(h)\im(\lambda_q)}{|\lambda_q|^2}\frac{|\lambda_q|}{\sqrt{V_{q}^2(h)+U_{q}^2(h)}}.\nonumber
\end{align}
From \eqref{uvh2} and \eqref{vuoh}, it follows that $\lim\limits_{\lambda_q\rightarrow0}\cos(\phi_{q}(h)-\theta_q)=1$.

By (a) and the continuity of the sine and cosine functions, there is only one limit point of $\phi_{q}(h)-\theta_q$, that is, 0, in the closure of $(-\pi, \pi)$ such that (b) and (c) hold. Therefore, $\lim\limits_{\lambda_q\rightarrow0}(\phi_{q}(h)-\theta_q)=0$.

\textbf{(3)} 
It is easy to verify that
\begin{align*}
\lim_{h\rightarrow0}U_{q}\left(h\right) & =0,\\
\lim_{h\rightarrow0}V_{q}\left(h\right) & =0,\\
\dot{U}_{q}\left(h\right) & =\re(\lambda_{q})e^{\re(\lambda_{q})h}+\im(\lambda_{q})\sin\left(\im(\lambda_{q})h\right),\\
\dot{V}_{q}\left(h\right) & =\im(\lambda_{q})\cos\left(\im(\lambda_{q})h\right).
\end{align*}
Then, by using the L'H\^{o}pital's rule and these equalities, we have that
\begin{alignat*}{1}
 & \lim_{h\rightarrow0}\frac{U_{q}^{2}\left(h\right)+V_{q}^{2}\left(h\right)}{h^{2}}\\
= & \lim_{h\rightarrow0}\frac{U_{q}\left(h\right)\dot{U}_{q}\left(h\right)+V_{q}\left(h\right)\dot{V}_{q}\left(h\right)}{h}\\
= & \lim_{h\rightarrow0}\left[U_{q}\left(h\right)\ddot{U}\left(h\right)+\dot{U}_{q}^{2}\left(h\right)+V_{q}\left(h\right)\ddot{V}\left(h\right)+\dot{V}_{q}^{2}\left(h\right)\right]\\
= & \lim_{h\rightarrow0}\left[\dot{U}_{q}^{2}\left(h\right)+\dot{V}_{q}^{2}\left(h\right)\right]
=  \left|\lambda_{q}\right|^{2}
\end{alignat*}

\textbf{(4)}
By following the same line as used in the proof of (a), we have that $(\phi_{q}(h)-\theta_q) \in (-\pi, \pi)$ when $h\rightarrow0$. Then, using the L'H\^{o}pital's rule gives that
\begin{align}\lim_{h \rightarrow0}\frac{V_{q}(h)}{U_{q}(h)}=&\lim_{h \rightarrow0}\frac{\sin(\im(\lambda_q)h)}{e^{\re(\lambda_q)h}-\cos(\im(\lambda_q)h)}\nonumber\\
=&\lim_{h \rightarrow0}\frac{\im(\lambda_q)\cos(\im(\lambda_q)h)}{\re(\lambda_q)e^{\re(\lambda_q)h}+\im(\lambda_q)\sin(\im(\lambda_q)h)}\nonumber\\
=&\frac{\im(\lambda_q)}{\re(\lambda_q)}=\tan(\theta_q),\nonumber
\end{align}
which together with $(\phi_{q}(h)-\theta_q) \in (-\pi, \pi)$ implies that $$\lim\limits_{h\rightarrow0}(\phi_{q}(h)-\theta_q)=0.$$

\section*{Appendix B}\label{appendixaa}
{\myr \begin{strip}\begin{proof}
Denote $J_S=PSP^{-1}$ and $J_H=QHQ^{-1}$, where $J_H$ and $J_S$ are the Jordan form of $H$ and $S$, respectively. Then, we have
\begin{align}
(P\otimes Q)F\left(\mu\right)(P\otimes Q)^{-1} &=(P\otimes Q)\big[I_N\otimes e^{Sh}-\mu H\otimes \int_{0}^{h}e^{S\tau}d\tau\big](P\otimes Q)\nonumber\\
&=(P\otimes Q)\big[I_N\otimes e^{Sh}-\mu H\otimes \int_{0}^{h}e^{S\tau}d\tau\big](P^{-1}\otimes Q^{-1})\nonumber\\
&=(P\otimes Q)(I_N\otimes e^{Sh})(P^{-1}\otimes Q^{-1})-\mu (P\otimes Q)\big[H\otimes \int_{0}^{h}e^{S\tau}d\tau\big](P^{-1}\otimes Q^{-1})\nonumber\\
&=\big[PP^{-1}\otimes Qe^{Sh}Q^{-1}-\mu PHP^{-1}\otimes \int_{0}^{h}Qe^{S\tau}Q^{-1}d\tau\big]\nonumber\\
&=I_N\otimes e^{J_Sh}-\mu J_H\otimes \int_{0}^{h}e^{J_S\tau}d\tau. \label{matrxdecomp}
\end{align}
The matrix \eqref{matrxdecomp} is the upper triangular matrix, provided that $J_H$ and $J_S$ are the Jordan form. Thus all the eigenvalues of $F\left(\mu\right)$ are in the main
diagonal, which completes the proof.
\end{proof}\end{strip}}

\section*{Appendix C}\label{appendixb}
$\alpha_{i,q}$, $\beta_{i,q}$, and $\gamma_{i,q}$ are evaluated as follows:
\begin{align}
\alpha_{i,q}
=&\frac{|\lambda_i|^2}{|\lambda_q|^2}\big|e^{\lambda_q h}-1\big|^2\nonumber\\
=&\frac{|\lambda_i|^2}{|\lambda_q|^2}\left(e^{2\re(\lambda_q)h}-2e^{\re(\lambda_q)h}\cos(\im(\lambda_q)h)+1\right)\nonumber\\
=&\frac{|\lambda_i|^2}{|\lambda_q|^2}\big(V_{q}^2(h)+U_{q}^2(h)\big),\nonumber\\
\beta_{i,q}=&-\int_{0}^{h}\big[\lambda_i e^{\lambda_q\tau+\lambda_q^*h}+\lambda_i^* e^{\lambda_q^*\tau+\lambda_qh}\big]d\tau\nonumber\\
=&|\lambda_i|\int_{0}^{h}\big[e^{\theta_i j}e^{\lambda_q\tau+\lambda_q^*h}+e^{-\theta_i j} e^{\lambda_q^*\tau+\lambda_qh}\big]d\tau\nonumber\\
=&-\frac{|\lambda_i|}{|\lambda_q|^2}\big[\lambda_q^{*}e^{\lambda_qh+\lambda_q^*h+\theta_i j}+ \lambda_qe^{\lambda_q^*h+\lambda_qh-\theta_i j}\nonumber\\
&-\lambda_q^{*}e^{\lambda_q^*h+\theta_i j}- \lambda_qe^{\lambda_qh-\theta_i j}\big]\nonumber\\
=&-\frac{|\lambda_i|}{|\lambda_q|}\big[e^{\lambda_qh+\lambda_q^*h+\theta_i j-\theta_q j}-e^{\lambda_q^*h+\theta_i j-\theta_q j}\nonumber\\
&+ e^{\lambda_q^*h+\lambda_qh-\theta_i j+\theta_q j}- e^{\lambda_qh-\theta_i j+\theta_q j}\big]\nonumber\\
=&-2\frac{|\lambda_i|}{|\lambda_q|}e^{\re(\lambda_q)h}\big[\cos(\theta_i-\theta_q)e^{\re(\lambda_q)h}\nonumber\\
&-\cos(\im(\lambda_q)h-\theta_i+\theta_q)\big]\nonumber\\
=&-2\frac{|\lambda_i|}{|\lambda_q|}e^{\re(\lambda_q)h}\big[\cos(\theta_i-\theta_q)U_{q}(h)\nonumber\\
&-V_{q}(h)\sin(\theta_i-\theta_q)\big]\nonumber\\
=&-2\frac{|\lambda_i|}{|\lambda_q|}e^{\re(\lambda_q)h}\sqrt{V_{q}^2(h)+U_{q}^2(h)}\cos(\psi_{i,q}(h)),\nonumber\\
\gamma_{i,q}=&e^{2\re(\lambda_q)h}-1,\nonumber
\end{align}
where $\psi_{i,q}(h)$, $U_{q}(h)$, and $V_{q}(h)$ are defined in \eqref{UVPT}.
Then, we have that
\begin{align}
\beta_{i,q}^2-&4\alpha_{i,q}\gamma_{i,q}\nonumber\\
=&4\frac{|\lambda_i|^2}{|\lambda_q|^2}e^{2\re(\lambda_q)h}(V_{q}^2(h)+U_{q}^2(h))\cos^2(\psi_{i,q}(h))\nonumber\\
&-4\frac{|\lambda_i|^2}{|\lambda_q|^2}\big(V_{q}^2(h)+U_{q}^2(h)\big)\left(e^{2\re(\lambda_q)h}-1\right)\nonumber\\
=&4\frac{|\lambda_i|^2}{|\lambda_q|^2}(V_{q}^2(h)+U_{q}^2(h))\big[1-e^{2\re(\lambda_q)h}\sin^2(\psi_{i,q}(h))\big].\nonumber
\end{align}

\bibliographystyle{IEEEtran}
\bibliography{myref}
\end{document}